\documentclass[runningheads]{llncs}
\usepackage[T1]{fontenc}
\usepackage{graphicx}
\usepackage[hidelinks]{hyperref}
\usepackage{color}

\usepackage{pgfplots}
\usepackage{orcidlink}
\pgfplotsset{compat=1.18}
\usepackage{pgfplotstable}
\usepgfplotslibrary{groupplots}
\usepackage{tikz}
\usepackage{subcaption}

\usepackage{caption}
\definecolor{lightred}{RGB}{255,100,100} 
\definecolor{darkred}{RGB}{150,50,50}
\definecolor{lightblue}{RGB}{100,100,255} 
\definecolor{lightgreen}{RGB}{100,255,100} 
\definecolor{darkblue}{RGB}{50,50,150} 
\definecolor{darkgreen}{RGB}{50,150,50} 
\definecolor{darkdarkblue}{RGB}{50,50,100} 
\definecolor{darkdarkgreen}{RGB}{50,100,50} 

\usepackage{url}
\usepackage{amsfonts}
\usepackage{amsmath}
\usepackage{multirow}
\usepackage{multicol}
\usepackage{adjustbox}
\usepackage{enumitem}
\usepackage[ruled]{algorithm}
\usepackage{float}
\usepackage[noend]{algpseudocode}
\newtheorem{lem}{Lemma}
\usepackage{todonotes}

\usepackage{xspace}

\newcommand{\setchain}{{Setchain}\xspace}
\newcommand{\PR}[1]{\ensuremath{\textit{#1}}}
\newcommand{\proofs}{\PR{proofs}\xspace}
\newcommand{\epoch}{\PR{epoch}\xspace}
\newcommand{\history}{\PR{history}\xspace}
\newcommand{\TheSet}{\PR{the}\_\PR{set}\xspace}

\newcommand{\batch}{\PR{batch}\xspace}
\newcommand{\batchOriginal}{\PR{batch}\_\PR{original}\xspace}
\newcommand{\hashToBatch}{\PR{hash}\_\PR{to}\_\PR{batch}\xspace}
\newcommand{\hashToSigners}{\PR{hash}\_\PR{to}\_\PR{signers}\xspace}

\newcommand{\API}[1]{\texttt{#1}\xspace}
\newcommand{\APPEND}{\API{Append}}
\newcommand{\NEWBLOCK}{\API{NewBlock}}

\newcommand{\BroadcastTxAsync}{\API{BroadcastTxAsync}}

\newcommand{\FinalizeBlock}{\API{FinalizeBlock}}

\newcommand{\validelement}{\ensuremath{\mathtt{valid\_element}}\xspace}

\newcommand{\setobject}{\ensuremath{\mathtt{S}}\xspace}
\newcommand{\add}{\ensuremath{\mathtt{add}}\xspace}
\newcommand{\addepoch}{\ensuremath{\mathtt{add\_to\_batch}}\xspace}
\newcommand{\get}{\ensuremath{\mathtt{get}}\xspace}
\newcommand{\ledgerobject}{\ensuremath{\mathtt{L}}\xspace}
\newcommand{\Append}{\ensuremath{\mathtt{append}}\xspace}
\newcommand{\NewBlock}{\ensuremath{\mathtt{new\_block}}\xspace}
\newcommand{\Sign}{\ensuremath{\mathtt{Sign}}\xspace}
\newcommand{\Hash}{\ensuremath{\mathtt{Hash}}\xspace}
\newcommand{\validproof}{\ensuremath{\mathtt{valid\_proof}}\xspace}

\algblockdefx[Upon]{Upon}{EndUpon}%
[1]{{\bf upon} (#1) do}%
{{\bf end upon}}

\begin{document}
\title{Setchain Algorithms for Blockchain Scalability\thanks{This work is funded in part by a research grant from Nomadic Labs and the Tezos Foundation, and by MICIU/AEI /10.13039/501100011033/, ERDF, and the ESF+ under predoctoral training grant PREP2022-000373, grant DRONAC (PID2022-140560OB-I00), and grant DECO (PID2022-138072OB-I00).}}


\author{Arivarasan Karmegam\inst{1,3}\orcidlink{0000-0002-6690-0285} \and
Gabina Luz Bianchi\inst{6} \and
Margarita Capretto\inst{2,4}\orcidlink{0000-0003-2329-3769} \and
Martín Ceresa\inst{5}\orcidlink{0000-0003-4691-5831} \and
Antonio {Fernández Anta}\inst{2,1}\orcidlink{0000-0001-6501-2377} \and
César Sánchez\inst{2}\orcidlink{0000-0003-3927-4773}}
\authorrunning{A. Karmegam et al.}
%
\institute{IMDEA Networks Institute, Madrid, Spain
\and
IMDEA Software Institute, Madrid, Spain
\and
Universidad Carlos III de Madrid, Spain
\and
Universidad Politécnica de Madrid, Spain
\and
Input-Output, Madrid, Spain
\and
Universidad Nacional de Rosario, Argentina }

\maketitle              
%

\begin{abstract}
Setchain has been proposed to increase blockchain scalability
by relaxing the strict total order requirement among transactions. 
Setchain organizes elements into a sequence of sets, referred to as \textit{epochs,} 
so that elements within each epoch are unordered.
%
In this paper, we propose and evaluate three distinct Setchain algorithms, that leverage an underlying block-based ledger.
\textit{Vanilla} is a basic implementation that serves as a reference point.
\textit{Compresschain} aggregates elements into batches, and compresses these batches before appending them as epochs in the ledger.
\textit{Hashchain} converts batches into fixed-length hashes which are appended as epochs in the ledger.
This requires Hashchain to use a
distributed service to obtain the batch contents from its hash.
%
To allow light clients to safely interact with only one server, 
the proposed algorithms maintain, as part of the Setchain, proofs for the epochs. An \textit{epoch-proof} is the hash of the epoch, cryptographically signed by a server. A client can verify the correctness of an epoch with $f+1$ epoch-proofs (where $f$ is the maximum number of Byzantine servers assumed).
All three Setchain algorithms are implemented on top of the CometBFT blockchain application platform. 
%
We conducted performance evaluations across various configurations, using clusters of four, seven, and ten servers. Our results show that the Setchain algorithms reach orders of magnitude higher throughput than the underlying blockchain, and achieve finality with latency below $4$ seconds.


\end{abstract}

\section{Introduction}

%
%
%
A blockchain is a \emph{reliable distributed object} containing the
list of transactions performed on behalf of the
users, totally ordered and packed into blocks~\cite{anta2018formalizing,anta2021principles}.
Real-world blockchains are maintained by multiple servers (without a central authority) that leverage a \emph{Byzantine-tolerant consensus algorithms} to order the transactions and
that compute their effects~\cite{DBLP:journals/tocs/CastroL02,redbelly}.

\textbf{Challenges.}
One of the main challenges for practical blockchains is to 
append transactions to the chain at a fast rate. This is called
\textit{throughput} and is measured in transactions per second (TPS).
Throughput is constrained by multiple factors, including the speed at which the consensus is solved, network latency, and the computational cost of executing transactions to update or validate the blockchain state. To maintain efficiency and decentralization, blockchains often impose limits on block size or execution complexity, which further impact throughput.
%
%
A second challenge is to reduce the \textit{latency} of the blockchain,
i.e. 
the time spent since users send their transactions until those transactions are added to the blockchain.
The latency in
some blockchains can be of hours. 
Moreover, in some blockchains
(e.g., Bitcoin), a block that has been reported to be added to the
blockchain may end up being removed due
to a fork in
the chain. Hence, a third challenge is to achieve \textit{finality}
in the blockchain, i.e., reach a point in time when users know
their transactions are in the chain to never be removed.
%
Finally, a fourth challenge is to provide users with \textit{proofs} that their transactions are in the chain.
This can be achieved by having each user contact several servers and ask them for individual proof (e.g., a copy of the whole chain). This is rarely done, and users typically trust that the server with which they interact is honest. It would be desirable that solid proofs that the transactions are appended to the chain are provided to users who interact with one single server.
We address these challenges with \setchain~\cite{capretto2024improving}, which is a \textit{reliable distributed object} that implements Byzantine-tolerant distributed grown-only sets with barriers and has been shown to have a large throughput, and with \textit{epoch-proofs,} that are cryptographic artifacts that allow users to interact reliably with just one server.
Related Works section is deferred to Appendix \ref{sec:relatedworks} due to space constraints.

\textbf{Setchain.}
An approach to improve the scalability of blockchain is \setchain~\cite{capretto2024improving},
a Byzantine-tolerant distributed object
that implements a sequence of sets (called \textit{epochs}).
%
\setchain relaxes the total order requirement of blockchain and thus can achieve higher throughput and scalability by allowing validating transactions in parallel within an epoch.
\setchain~can be used for applications like digital
registries (e.g., MIT Digital Diplomas~\cite{blockcerts,mitblockchaindiploma}, Georgia Government registry) or voting systems (e.g., Follow My Vote, 
Chirotonia~\cite{DBLP:conf/blockchain2/RussoAVR21}), where different elements in the blockchain need not be ordered except across infrequent epoch barriers.
Byzantine-tolerant distributed algorithms that implement \setchain have been proposed, but no efficient real-world implementation exists.

\textbf{Contributions.}
In this work we propose a family
of real-world \setchain~algorithms built on top of a block-based ledger. We follow an incremental approach, providing several approximations to the
final and most complex solution (Section~\ref{sec:algorithms}). We first present \textbf{Algorithm Vanilla} which is a naive implementation of \setchain with similar throughput as the underlying ledger. Then, we present \textbf{Algorithm Compresschain} that increases the throughput using compression of epochs. Finally, we present \textbf{Algorithm Hashchain}, a solution using hash functions.
From these algorithms, our primary contribution is 
Hashchain, which exploits the succinctness of hashes to reduce the communication necessary during broadcasts
and consensus, communicating a fixed-size hash instead of the hundreds or thousands of
elements of an epoch.
The price to pay is an additional distributed algorithm to obtain the epoch contents of a hash from the corresponding server.
Additionally, to allow users to interact with only one server, these algorithms maintain, as part of the Setchain, proofs for its epochs. An \textit{epoch-proof} is the hash of the epoch, cryptographically signed by a server. A user can verify the correctness of an epoch with $f+1$ epoch-proofs (where $f$ is the maximum number of Byzantine servers).
All three Setchain algorithms are implemented on top of the CometBFT blockchain application platform (Section~\ref{sec:implementation}). We deployed these implementations as docker nodes in a cluster.
%
We conducted performance evaluations across various configurations, using clusters of four, seven, and ten servers. Our results show that the Setchain algorithms reach orders of magnitude higher throughput than the underlying blockchain (tens of thousands TPS), and achieve finality with latency below $4$ seconds.




%

\section{Model and Definitions} \label{sec:modelsanddefinition}

\textbf{System Model.}\label{subsec:systemmodel}
We consider a distributed system consisting of $n$ servers and an
unbounded number of clients, which together are the system processes.
%
The system is permissioned but public, also known as ``open
permissioned''~\cite{redbelly}, referring to an open model for clients
but where servers are known upfront (permissioned).
Nevertheless, this model can also be adapted to a permissionless
setting with committee sortition~\cite{gilad2017algorand} without
significant modifications.
%
%
Up to $f<n/2$ of the servers may exhibit Byzantine behavior, while the
clients' behavior is not restricted (i.e., they can be Byzantine). We
assume that the number $f$ of Byzantine servers is known.
%
We use \(f+1\) as a lower bound on the amount of consistent
epoch-proof required to ensure that an epoch is correct (see
Section~\ref{subsec:epochproofs}) and as the number of signatures
needed in order to consolidate hashes into epochs (see
Section~\ref{subsec:hashchain}).
We assume that there is a public key infrastructure (PKI) deployed so that each process (server or client) has a pair of private and public keys. Every process also knows the public key of all the other processes. The communication between processes is assumed to be reliable, meaning that messages sent between correct processes are eventually delivered only once, and no spurious messages are generated. Faulty processes can send any message but thanks to the PKI they cannot impersonate other processes. The PKI enables authenticated and secure communication between processes, since messages are signed by the sender using their private key, and the receiver verifies the signature using the sender's public key. If a signature is invalid, the message is discarded.
Clients create elements and invoke a \setchain operation to add
them. The elements created by the clients are signed and can hence be
authenticated. They can also be validated by the servers for syntactic
and semantic correctness. Only authenticated valid elements are
processed by the correct servers. We assume that a server cannot
create a valid element by itself, and that clients and servers do not collude.

\textbf{\setchain.}\label{subsec:setchain}
A \setchain~\cite{capretto2024improving} is a Byzantine-tolerant distributed object
$\setobject$
that implements a sequence of sets called \textit{epochs}.
The \setchain defines an order between elements in different epochs but not between
elements in the same epoch.
Thus, \setchain relaxes the total order requirement imposed by blockchains, potentially
achieving higher throughput and scalability.
%
%
Let \(U\) be the set of elements that client processes can inject into the
\setchain.
A \setchain \setobject is a reliable distributed object where a collection of servers
maintain:
(i) a grow-only set $\TheSet \subseteq U$ of elements added;
(ii) a natural number $\epoch \in \mathbb{N}$;
(iii)
a map $history: \{1, \ldots, \epoch\} \rightarrow 2^U$ with the set of elements that have been stamped with a given
epoch number\footnote{$2^U$ denotes the
power set of $U$.};
(iiii)
a set $\proofs$ of epoch-proofs (see Section~\ref{subsec:epochproofs})\footnote{This set was not part of the API described in the original \setchain paper~\cite{capretto2024improving} and is one of the contributions of this work.}.
%
Servers support two operations: \(\add\) and \(\get\).
Operation \(\setobject.\add_v(e)\) is used by a client to request a
server $v$ to add an element $e$ to the \setchain \setobject, while
operation \(\setobject.\get_v()\) returns the values
$(\TheSet,$ $\history,\epoch, \proofs)$ that server $v$ maintains.  Servers
can use a function $\validelement(e)$ to locally decide if an element
$e$ is valid.
%
%
%
In a \setchain, when a new epoch is created, the servers collaboratively
decide which of the added elements without an epoch assigned are included in the new epoch, 
and increase the epoch number.
We call this an \emph{epoch increment}.
%
For the rest of the paper, we assume that at any given time, there will be a future epoch increment. This assumption is reasonable in real-world scenarios, and it can be reliably ensured by utilizing timeouts. 
A typical workflow from the point of the view of a client is as follows: a
client invokes $\setobject.\add_v(e)$ in a server $v$ to insert a new valid
element $e$ in the \setchain.
The element $e$ will be propagated among the servers, and when an
epoch increment occurs, the servers will attempt to include it in the
new epoch.
After waiting for some time, the client invokes $\setobject.\get_w()$
in a (possibly different) server $w$ to check that the element has
been effectively added to the \setchain and included in an epoch.
%
%
\setchain implementations must satisfy certain properties that guarantee consistency between correct servers and that the valid elements added are eventually included in an epoch~\cite{capretto2024improving}. These properties are restricted to correct servers, since Byzantine servers do not provide any guarantee. 

\begin{enumerate}
\item
\textit{Consistent-Sets:} \label{prop:consistent-sets}
Let $(T,H,h,P)=\setobject.\get_v()$ be the result of an invocation to a correct server $v$. Then,  for all $i \in \{1,\ldots,h\}, H[i] \subseteq T$.
\item
\textit{Add-Get-Local:} \label{prop:add-get-local}
Let $\setobject.\add_v(e)$ be an operation invoked in a correct server $v$ and let $e$ be valid. Then, eventually all invocations $(T,H,h,P)=\setobject.\get_v()$ satisfy $e\in T$.    
\item
\textit{Get-Global:} \label{prop:get-global}
Let $v$ and $w$ be two correct servers, let $e$ be a valid element, and let $(T,H,h,P)=\setobject.\get_v()$. If $e \in T$, then eventually all invocations $(T',H',h',P')=\setobject.\get_w()$ satisfy that $e \in T'$.    
\item
\textit{Eventual-Get:} \label{prop:eventual-get}
    Let $v$ be a correct server, let $e$ be a valid element and let $(T,H,h,P)=\setobject.\get_v()$. If $e \in T$, then eventually all invocations $(T',H',h',P')\\=\setobject.\get_v()$ satisfy that $e \in H'$.
\item
\textit{Unique-Epoch:} \label{prop:unique-epoch}
    Let $v$ be a correct server, $(T,H,h,P)=\setobject.\get_v()$, and let $i,i' \in \{1,\ldots,h\}$ with $i \neq i'$. Then, $H[i] \cap H[i'] = \emptyset$.
\item
\textit{Consistent-Gets:} \label{prop:consistent-gets}
    Let $v,w$ be correct servers, let $(T,H,h,P)=\setobject.\get_v()$ and $(T',H',h',P')=\setobject.\get_w()$, and let $i \in \{1,\ldots,$ $\min(h,h')\}$. Then $H[i]=H'[i]$.
\item
\textit{Add-before-Get:} \label{prop:add-before-get}
  Let $v$ be a correct server, $(T,H,h,P)=\setobject.\get_v()$, and
  $e \in T$ a valid element. Then there was an operation
  $\setobject.\add_w(e)$ invoked in the past in some server $w$.
\end{enumerate}


\setcounter{property}{7}


Properties \ref{prop:consistent-sets}, \ref{prop:unique-epoch}, \ref{prop:consistent-gets}, and \ref{prop:add-before-get} are safety properties. Properties \ref{prop:add-get-local}, \ref{prop:get-global}, and \ref{prop:eventual-get} are liveness properties.

\textbf{\setchain Epoch-proofs.} \label{subsec:epochproofs}
The original \setchain~\cite{capretto2024improving} proposal assumes that a client interacts with a correct server. However, in practice, clients do not know whether they are communicating with a correct or a Byzantine server. Hence, 
a client may need to interact with a sufficient number of servers (at least $f+1$) to guarantee that at least one is correct.
In this work, we introduce \emph{epoch-proofs} as a mechanism that allows clients to achieve the same guarantees without needing to contact multiple servers, thereby simplifying the process and improving efficiency.
%
%
An epoch-proof is the cryptographic signature of an epoch $i$ by a server $v$. As described above, in this paper, the \setchain maintains a set $\proofs$ of epoch-proofs. This set is returned by the \setchain when a \(\get\) operation is invoked. Hence, when $\proofs$ contains at least $f+1$ consistent epoch-proofs for a given epoch $i$, the client can be sure the epoch is correct.
The epoch-proof is created by signing the hash of the epoch number and the elements of the epoch: $p_v(i) = \Sign_v(\Hash(i, \history[i]))$.
%
%
To add an element, a client only performs a single \(\setobject.\add_v(e)\) request to one server $v$, hoping it is a
correct server.
After waiting for some time, the client can invoke a
\(\setobject.\get_w()\) from a single server $w$ and check whether $e$
is in some epoch and the set $\proofs$ returned contains at least
$f+1$ valid epoch-proofs for that epoch to guarantee that at least one
correct server has signed it.  Clients can verify if epoch
proofs are valid by generating the hash of a given epoch, and
verifying if the signature in the epoch-proof is valid using the hash
and the public key of the signing server.  Recall that servers' public
keys are known to the clients.
Note that this process usually requires only one message per \(\add\) and one message per \(\get\). Of course, it is possible that the client is unlucky
and servers $v$ or $w$ are Byzantine, and after waiting a reasonable amount of time, has to restart the process.
The following is the basic property that any \setchain algorithm must satisfy concerning epoch-proofs.
\begin{property}[Valid-Epoch]\label{prop:valid-epoch}
Let $v$ be a correct server, $(T,H,h,P)=\setobject.\get_v()$, and $i \in \{1, \ldots, h\}$. Then eventually all invocations $(T',H',$ $h',P')=\setobject.\get_v()$ satisfy that $P'$ contains at least $f+1$ epoch proofs of $H[i]$.
\end{property}


\textbf{Block-based Ledger.} \label{subsec:block-basedledger}
In the proposed algorithms to implement a \setchain presented in this paper, we assume the availability of a Byzantine-tolerant consensus service. 
For simplicity, we abstract that service as a block-based ledger. While our algorithms only require that the number of Byzantine processes satisfies $f<n/2$,
the block-based ledger may have more stringent requirements. For instance, in Section \ref{sec:implementation} we use the ledger CometBFT that requires $f<n/3$.
%
%
A block-based ledger $\ledgerobject$ is a Byzantine-tolerant distributed object that maintains a sequence of blocks. Each block contains a sequence of transactions that have been appended to the ledger by its clients. We prefer not to call this object a blockchain since its transactions have no semantics.
To keep the nomenclature consistent, the term \emph{transaction} is always used to
refer to the \emph{block-based ledger transactions}, while \emph{element} refers to the elements
of the \setchain.
Depending on the implementation, a single ledger transaction may contain one
or many elements.
%
%
The block-based ledger provides two endpoints. First,
$\Append(tx)$ is used to submit a transaction $tx$ into the ledger,
which eventually gets included in a block. Then, $\NewBlock(B)$ notifies the servers
whenever a new block $B$ has been appended.  The number of
transactions in $B$ is $|B|$, and the $i$th transaction in $B$ is
$B[i]$.
%
%
Since \setchain algorithms use a block-based ledger, to prove correctness it is necessary to define properties that the block-based ledger guarantees. In particular, the ledger property that will be used is that any valid transaction appended by a correct server will eventually be included in a block, which gets notified to all the servers. Moreover, we assume that all blocks notified are final.
\begin{property}[Ledger-Add-Eventual-Notify]\label{prop:ledger-add-eventual-notify}
Let $tx$ be a valid transaction, and let $v$ be a correct server. Then, if $v$ invokes $\ledgerobject.\Append_v(tx)$, transaction $tx$ will be eventually and permanently added in a fixed position $i$ into a block $B$.
Moreover, all correct servers $w$ will be eventually notified of the new block $B$ with $\ledgerobject.\NewBlock_w(B)$.
\end{property}

\begin{property}[Ledger-Consistent-Notification]\label{prop:ledger-consistent-notification}
All correct servers $w$ are notified with $\ledgerobject.\NewBlock_w(B)$ of the same set of blocks, and in the same order.
\end{property}

\begin{property}[Notification-Implies-Append]\label{prop:notification-implies-append}
If a correct server $w$ is notified with $\ledgerobject.\NewBlock_w(B)$ containing some valid element $e$, then some server $v$ had invoked $\ledgerobject.\Append_v(e)$.
\end{property}

\section{Setchain Algorithms}\label{sec:algorithms}

In this section we propose practical \setchain algorithms built on top
of a block-based ledger.
In particular, we present three different implementations, beginning with a naive
but trivially correct implementation and ending with a more complex algorithm
implementing \setchain using hashing.
%
%
%
As described in Section~\ref{subsec:setchain}, a \setchain \setobject provides two methods, \(\add\) and \(\get\), defined in each algorithm.

\begin{figure}[h]
  \begin{adjustbox}{minipage=[t]{\columnwidth}}
  \algrenewcommand\alglinenumber[1]{\scriptsize #1}
    \begin{algorithm}[H]
      \renewcommand{\thealgorithm}{Compresschain}         
      \caption{\small 
      Code executed by server $v$.
      Uses block-based ledger $\ledgerobject$ shared by all servers.
}%
      \label{alg:setchain-compress}%
      \scriptsize
      \begin{multicols}{2}
      \begin{algorithmic}[1]
      \vspace*{-2.5em}
      \algrenewcommand\algorithmicindent{1em}
      \State Init: $\TheSet \leftarrow \emptyset$, $\epoch \leftarrow 0$, $\history \leftarrow \emptyset$, 
      \Statex \hspace*{2em} $\proofs \leftarrow \emptyset$, $\batch \leftarrow \emptyset$
      \Function{\add}{$e$}
        \State assert $\validelement(e) \land (e \notin \TheSet)$ 
        \label{alg:setchain-compress:valid_element}
        \State $the\_set \leftarrow \TheSet \cup \{ e \}$ \label{alg:setchain-compress-setaddition}
        \State $\addepoch(e)$ \label{alg:setchain-compress-batchaddition}
        \State \textbf{return}
      \EndFunction

        \Function{\addepoch}{$e$}
        \State $\batch \leftarrow \batch \cup \{ e \}$
        \State \textbf{return}
        \EndFunction
      
        \Function{\get}{\null}
            \State \textbf{return} $(\TheSet, \history, \epoch, \proofs)$
        \EndFunction

        \Upon{$\mathtt{isReady}(\batch)$}
            \State assert $\batch \neq \emptyset$
            \State $b \leftarrow \mathtt{Compress}(\batch)$
            \State $\ledgerobject.\Append(b)$ \label{alg:setchain-compress-append}
            \State $\batch \leftarrow \emptyset$\label{alg:setchain-compress-emptybatch}
        \EndUpon
        
        \Upon{$\ledgerobject.\NewBlock(B)$}

             \For{$i=1$ to $|B|$}
                    \State $\batchOriginal \leftarrow \mathtt{Decompress}(B[i])$
                    \If{$\batchOriginal = \emptyset$}
                        \textbf{continue}
                    \EndIf
                    \State $np \leftarrow \{ep \in \batchOriginal: $
                    \Statex \hspace*{3em} $ep=\langle j, p, w \rangle \text{~is an epoch-proof}$ 
                    \Statex \hspace*{3em} $\land \  \validproof(j, p, w, \history[j])\}$
                    \State $\proofs \leftarrow \proofs \cup np$ \label{alg:setchain-compress:validepochproof}
                    \State $G \leftarrow \{e \in \batchOriginal:$
                    \Statex \hspace*{3em} $e \text{~is an element} \land \validelement(e)$
                    \Statex \hspace*{3em} $\land (e \notin \history)\}$ \label{alg:setchain-compress-othervalidelement}
                    \State $\TheSet \leftarrow \TheSet \cup G$ \label{alg:setchain-compress-othersetaddition}
                    \State $\epoch \leftarrow \epoch + 1$ \label{alg:setchain-compress-epochinc}
                    \State $\history[\epoch] \leftarrow G$ \label{alg:setchain-compress-epochaddition}
                    \State $p \leftarrow \Sign_v(\Hash(\epoch,G))$
                    \State $\addepoch(\langle \epoch, p, v \rangle)$ \label{alg:setchain-compress:epochproof}
            \EndFor
        \EndUpon

        \end{algorithmic}
        \end{multicols}
        \vspace*{-1em}
      \end{algorithm}
	\end{adjustbox}
  \end{figure}

\textbf{\ref{alg:setchain-vanilla}.}\label{subsec:vanilla}
In ~\ref{alg:setchain-vanilla} we propose a basic implementation of \setchain \setobject
using a block-based ledger $\ledgerobject$.
The algorithm for \ref{alg:setchain-vanilla} is deferred to Appendix \ref{sec:algorithmvanilla} due to space constraints.
In this algorithm, each server $v$ maintains three local sets $\TheSet$, $\history$, and $\proofs$, and the counter $\epoch$, as defined in Section \ref{subsec:setchain}.
Whenever a client queries \setchain \setobject using $\setobject.\get_v()$, the server returns the current value of $\TheSet$, $\history$, $\epoch$, and $\proofs$ (Line~\ref{alg:setchain-vanilla-get}).
On the other hand, clients add an element $e$ by invoking $\setobject.\add_v(e)$. The server only accepts valid elements not present in $\TheSet$ (see Line~\ref{alg:setchain-vanilla:valid_element}). If so, the element is added to $\TheSet$, and the server invokes a $\ledgerobject.\Append_v(e)$ operation, which eventually adds the element $e$ to ledger $\ledgerobject$.
Whenever a new block gets appended to the ledger \ledgerobject, server $v$ is notified through $\ledgerobject.\NewBlock(B)$ (see Line~\ref{alg:vanilla-newblock}).
Ledger \ledgerobject is used to disseminate the epoch-proofs in the three algorithms proposed.
Hence, the server starts by extracting the valid epoch-proofs from $B$ and adding them to the set $\proofs$ (Line~\ref{alg:setchain-vanilla:validepochproof}). 
Then, the server extracts to a batch $G$ the elements from $B$ that are valid and are not in an epoch yet. (Observe that checking if an element is valid cannot be avoided because a Byzantine server may have added invalid elements to the ledger). The set $G$ of valid elements is added to $\TheSet$, and forms a new epoch. This new epoch is added to the $\history$ map. Then, the epoch-proof containing the epoch number, the epoch signature, and the identity of the server is added to the ledger by invoking an $\ledgerobject.\Append$ operation (Line~\ref{alg:setchain-vanilla:epochproof}).
~\ref{alg:setchain-vanilla} solves the proof challenge via epoch-proofs, allowing a client to interact with only one correct server. However, the throughput and latency of the \setchain implemented are those of the block-based ledger $\ledgerobject$.

\textbf{\ref{alg:setchain-compress}.}\label{subsec:compresschain}
%
%
%
%
The second algorithm we propose, \ref{alg:setchain-compress}, increases the throughput with respect to \ref{alg:setchain-vanilla}.
%
In \ref{alg:setchain-compress} each server has a \textit{collector} in which the elements added by the clients and the epoch-proofs added by the servers are held until the collector size is reached (or possibly a timeout expires). Then, the collected batch, which will become a new epoch, is compressed and appended to the ledger as a single ledger transaction.
\ref{alg:setchain-compress} adapts \ref{alg:setchain-vanilla} to handle the compressed batches, by adding an additional set $\batch$ to collect the client elements and the epoch-proof. When a client adds an element $e$ to the \setchain using $\setobject.\add_v(e)$, on top of adding the element to $\TheSet$, it is also added to $\batch$ (the collector; Line~\ref{alg:setchain-compress-batchaddition}). When $\batch$ reaches the collector size (or a timeout is triggered, and $\batch$ is not empty) a notification $\mathtt{isReady}(\batch)$ is signaled. Then, the batch is compressed and appended to the ledger, and $\batch$ is reset to $\emptyset$. 
When the creation of a new block $B$ is notified, each compressed batch in $B$ is processed in order. Let $B[i]$ be the $i$th transaction in block $B$. $B[i]$ is decompressed, and the valid epoch-proofs in $B[i]$ are added to the set $\proofs$. Then, the valid elements in $B[i]$ are extracted to a set $G$ to form a new epoch. The elements in $G$ are first added to $\TheSet$. Then, they are assigned an epoch number, and the epoch information is added to the $\history$ map. An epoch-proof for the epoch is also created and sent to the collector using $\setobject.\addepoch_v(\cdot)$ which in turn adds it to $\batch$. 
The most important difference between \ref{alg:setchain-compress} and \ref{alg:setchain-vanilla} is that in \ref{alg:setchain-compress} each transaction in a block $B$ is a compressed batch (which potentially has many Setchain elements) that becomes an epoch, whereas in \ref{alg:setchain-vanilla} each transaction in block $B$ is an element, and an epoch is the set of valid elements in the block $B$. This potentially increases the throughput significantly.

\begin{figure}[t!]
  \begin{adjustbox}{minipage=[t]{\columnwidth}} 
  \algrenewcommand\alglinenumber[1]{\scriptsize #1}
    \begin{algorithm}[H]
      \renewcommand{\thealgorithm}{Hashchain}         
      \caption{\small 
      Code executed by server $v$. Uses block-based ledger $\ledgerobject$ shared by all servers.}\label{alg:setchain-hash}
      \scriptsize
      \begin{multicols}{2}
      \begin{algorithmic}[1]
      \vspace*{-2.5em}
      \algrenewcommand\algorithmicindent{0.5em}
      \State Init: $\TheSet \leftarrow \emptyset, epoch \leftarrow 0, history \leftarrow \emptyset$, 
      \Statex \hspace*{1em} $\proofs \leftarrow \emptyset, \hashToBatch \leftarrow \emptyset$, 
      \Statex \hspace*{1em} $\batch \leftarrow \emptyset, \hashToSigners \leftarrow \emptyset$ 
      \Function{\add}{$e$}
        \State assert $\validelement(e) \land (e \notin \TheSet)$ 
        \label{alg:setchain-hash-validelement}
            \State $\TheSet \leftarrow \TheSet \cup \{ e \}$ \label{alg:setchain-hash-setaddition}
        \State $\addepoch(e)$ \label{alg:setchain-hash-batchaddition}
        \State \textbf{return}
      \EndFunction
      
        \Function{\addepoch}{$e$}
        \State $\batch \leftarrow \batch \cup \{ e \}$
        \State \textbf{return}
        \EndFunction

        \Function{\get}{\null}
            \State \textbf{return} $(\TheSet, history, epoch, \proofs)$
        \EndFunction
        
        \Upon{$\mathtt{isReady}(\batch)$} \label{alg:setchain-hash-isready}
            \State assert $\batch \neq \emptyset$
            \State $h \leftarrow \Hash(\batch)$
            \State $\hashToBatch[h] \leftarrow \batch$ \label{alg:setchain-hash-hashtobatch}
            \State $\mathtt{Register\_batch}(h,\batch)$ \label{alg:register-hash-hashtobatch}

            \State $s \leftarrow \Sign_v(h)$
            \State $hb \leftarrow \langle h, s, v \rangle$
            \State $\ledgerobject.\Append(hb)$ \label{alg:setchain-hash-append}
            \State $\batch \leftarrow \emptyset$
        \EndUpon
        
        \Upon{$\ledgerobject.\NewBlock(B)$}
       
             \For{$i=1$ to $|B|$}
                \If{ ($B[i] = \langle h, s_w, w \rangle) \land$ \newline 
                \hspace*{2em} $\mathtt{valid\_hash}(h,s_w,w)$}
                    \State $\batchOriginal \leftarrow \hashToBatch[h]$ \label{alg:setchain-hash-retrieveoriginalbatchlocal}
                    \If{$\batchOriginal = \emptyset$} \Comment{$h$ is new}
                    \State $\batchOriginal \leftarrow \mathtt{Request\_batch}(h)$  \label{alg:setchain-hash-requestoriginalbatch}
                    \If{$(\batchOriginal = \emptyset) \  \lor$ \Comment{Not found} \newline 
                    \hspace*{3em} $(\Hash(\batchOriginal) \neq h$)} \label{alg:setchain-hash-batchverify}
                        \State \textbf{continue}                    
                    \EndIf
                        \State $\hashToBatch[h] \leftarrow \batchOriginal$
                        \State $\mathtt{Register\_batch}(h,\batchOriginal)$\label{alg:other-register-hash-hashtobatch}
                        \State $s_v \leftarrow \Sign_v(h)$
                        \State $hb \leftarrow \langle h, s_v, v \rangle$
                        \State $\ledgerobject.\Append(hb)$ \label{alg:setchain-hash-otherappend}

                        \State $np \leftarrow \{ep \in \batchOriginal:$
                        \Statex \hspace*{3em} $ep=\langle j, p, w \rangle \text{~is epoch-proof~}  \land$
                        \Statex \hspace*{3em} $ \validproof(j, p, w, \history[j])\}$
                        \State $\proofs \leftarrow \proofs \cup np$ \label{alg:setchain-hash:validepochproof}
                        \State $G \leftarrow \{e \in \batchOriginal:$ \newline
                        \hspace*{3em} $ e \text{~is an element} \land \validelement(e)$ \newline
                        \hspace*{3em} $\land (e \notin history)\}$
                            \State $\TheSet \leftarrow \TheSet \cup G$
                            \label{alg:setchain-hash-othersetaddition}
                \EndIf
                    \State $\hashToSigners[h] \leftarrow \hashToSigners[h]$ \newline
                    \hspace*{3em} $ \cup \{w\}$ \label{alg:hash-consol}
                    \If{
                    $|\hashToSigners[h]| = f+1$}\label{alg:setchain-hash-signerscheck}
                        \State $epoch \leftarrow epoch + 1$ \label{alg:setchain-hash-epochinc}
                        \State $G \leftarrow \{e \in \batchOriginal:$ \newline
                        \hspace*{3em} $ e \text{~is an element}\land \validelement(e)$ \newline
                        \hspace*{3em} $ \land (e \notin history)\}$\label{alg:setchain-hash-othervalidelement}
                        \State $history[epoch] \leftarrow G$\label{alg:setchain-hash-epochaddition}
                        \State $p \leftarrow \Sign_v(\Hash(\epoch,G))$
                        \State $\addepoch(\langle epoch, p, v \rangle)$\label{alg:setchain-hash:epochproof}
                    \EndIf
                \EndIf
            \EndFor
        \EndUpon


        
        \end{algorithmic}
        \end{multicols}
        \vspace*{-1em}
      \end{algorithm}
	\end{adjustbox}
\end{figure}

\textbf{\ref{alg:setchain-hash}.}\label{subsec:hashchain} \ref{alg:setchain-hash} increases the throughput by using hashing of batches instead of compressing them.
While the space reduction of hashing may be enormous, hashes are
irreversible. Hence, a non-trivial method has to be provided to
recover the contents of the original batch of elements from a hash.
In \ref{alg:setchain-hash},
when a batch is ready (Line~\ref{alg:setchain-hash-isready}) server $v$ generates the hash $h$ of the batch, signs the hash $h$, and creates the hash-batch $hb$ as the combination of the hash of the batch, the signature, and its identity. Then, the hash-batch $hb$ is appended as a transaction to ledger $\ledgerobject$.
As the original batch cannot be recovered from the hash $h$, the server saves the batch associated with $h$ in a map $\hashToBatch$. Also, the \batch with its hash $h$ is registered using $\mathtt{Register\_batch}$. This will allow serving the contents of the original batch to other servers. 
Set $\batch$ is then reset to $\emptyset$.
Upon a $\ledgerobject.\NewBlock(B)$ notification, each valid
hash-batch $hb_w=\langle h, s_w, w \rangle \in B$ is processed as
follows. If server $v$ does not have the batch associated with the hash
$h$ in the map $\hashToBatch$, $v$ requests the batch to server $w$ using $\mathtt{Request\_batch}$
(which must have it since $w$ signed that hash-batch). 
Since $w$ may be Byzantine, server $v$ waits for a limited amount of time for the answer to its request. If the batch is
received correctly, $v$ creates its hash-batch
$hb_v=\langle h, s_v, v \rangle$ and appends it to ledger
$\ledgerobject$. Then, the batch is saved in map $\hashToBatch$ and
$\TheSet$ is updated.
Observe that finding a valid hash-batch in ledger $\ledgerobject$ is not enough to assign it an epoch number, because it could be a hash-batch appended by a Byzantine server that refuses to provide the batch that corresponds to the
hash. Hence, we decided that a hash has to be signed by at least $f+1$ individual servers to be consolidated into an epoch. The reason is that a hash with hash-batches from $f+1$ different servers is guaranteed to be signed by at least one correct server, which will have a copy of the batch and serve it upon request. 
%
Then, the process of hash-batch $hb_w$ is continued by adding $w$ to
the set of signers in the map $\hashToSigners$, used to track the set
of servers that appended hash-batches with hash $h$
(Line~\ref{alg:hash-consol}). When this set of signers reaches $f+1$,
the server extracts the valid elements from the batch to a set $G$
which is assigned an epoch number. We call this process
\textit{epoch consolidation}. Also, the valid epoch-proofs in \batchOriginal
are added to the set $\proofs$. As usual, the epoch is added to the
$\history$, and an epoch-proof for this epoch is generated and
sent to the collector using $\setobject.\addepoch_v(ep)$. 
\textbf{Correctness Proof Sketch.}
\label{sec:intuition-correctness}
We provide here an intuitive explanation for why the \setchain algorithms presented in Section \ref{sec:algorithms} behave correctly, even in the presence of up to $f$ Byzantine faults\footnote{We defer the full proof of correctness to Appendix \ref{sec:proofofcorrectness} due to space restrictions.}. The~\setchain algorithms ensure that all valid elements are eventually included in an epoch, and that all correct servers maintain a consistent view of this sequence. 
%
When a client submits an element to a~\setchain server $v$, it is added to that server's local $\TheSet$ only if the element is valid. Once stored locally, the element is returned as a part of that~\setchain server's~$\TheSet$, whenever $\setobject.\get_v()$ is invoked by the client (Property~\ref{prop:add-get-local}).
When it consolidates an epoch, a server adds all valid elements from the epoch to $\TheSet$, if not already present, before recording the epoch in the $history$. Therefore, $history$, which is the collection of all the epochs, is a subset of $\TheSet$ (Property \ref{prop:consistent-sets}). 
Each server periodically creates a batch either after collecting enough elements or after a timeout. Since all valid elements in the batch are added to $\TheSet$ and eventually put into an epoch, these elements eventually end up in $history$, again by Properties \ref{prop:ledger-add-eventual-notify} and \ref{prop:ledger-consistent-notification} (Property \ref{prop:eventual-get}).
%
In \ref{alg:setchain-hash}, a hash-batch consolidates into an epoch (i.e., is assigned an epoch number) only after the server receives $f+1$ signatures on the hash of the batch. When a correct server $v$ either generates a hash-batch or receives a hash-batch from another server, it first verifies the batch elements (requesting them to the signer of that batch, if the original batch is not found locally), then signs the batch hash and appends it to the ledger. Eventually, a block containing this hash-batch signed by $v$ is notified to all the servers, and the other correct servers request the original batch from $v$ (if the original batch is not found locally), verify the batch elements and sign the batch hash and append it to the ledger. So, given that there at least $f+1$ correct servers, eventually this hash-batch receives $f+1$ signatures and the batch consolidates to an epoch.
All correct~\setchain servers use the same logic to construct epochs from observed blocks, all~\setchain servers observe the same blocks in the same order, and all the transactions inside the block are totally ordered (Properties \ref{prop:ledger-add-eventual-notify} and \ref{prop:ledger-consistent-notification}).
Hence, correct~\setchain servers create identical epochs and agree on the epoch content (Property~\ref{prop:consistent-gets}). Also, each correct server ensures that new epochs exclude elements already included in $history$, which leads directly to Property~\ref{prop:unique-epoch}.
Again, by Property \ref{prop:ledger-add-eventual-notify}, every element added to a server's local $\TheSet$ eventually appears in a block that is delivered to all~\setchain servers. After that, every $\setobject.\get()$ invocation to a correct~\setchain servers will include this element as part of their~$\TheSet$ (Property \ref{prop:get-global}). Since we assume that a server can neither create a valid element by itself nor collude with clients, a server cannot append a valid element with $\ledgerobject.\Append(e)$ without a client invocation of $\setobject.\add(e)$ (Property \ref{prop:add-before-get}).
To ensure that clients can validate the correctness of an epoch without needing to trust a single server, Setchain maintains epoch-proofs, which are \setchain signatures of the hash of the epoch. These proofs are included in the ledger as transactions, enabling any client to retrieve them through a $\setobject.\get()$ operation. When a correct~\setchain server encounters an epoch, after validation, it signs the hash of the epoch and appends it to the block-based ledger as its epoch-proof. Due to Property \ref{prop:ledger-add-eventual-notify}, this epoch would have been notified to all~\setchain servers and eventually the ledger contains $f+1$ epoch-proofs for the epoch, as the number of correct servers is more than $f+1$. Therefore, eventually, when a client invokes $\setobject.\get()$, at least $f+1$ epoch-proofs are returned for an epoch (Property \ref{prop:valid-epoch}) and the client can trust the epoch.


\section{Implementation and Performance Evaluation}\label{sec:implementation}

The three algorithms presented have been implemented in Golang using CometBFT as the block-based ledger. We describe some implementation aspects.

\textbf{CometBFT.}
CometBFT \cite{tendermint.design} (previously known as Tendermint) is a Byzantine-tolerant state machine replication engine. CometBFT is a blockchain middleware that supports replicating arbitrary applications, written in any programming language. More details about CometBFT is given in Appendix \ref{sec:cometbft}.

\textbf{Performance Evaluation Platform.}
We carried out the performance evaluation of the \setchain using a cluster.
Each machine in the cluster has an Intel(R) Xeon(R) E-2186G CPU @ 3.80GHz with 12 cores, 32GB RAM, and runs Debian GNU/Linux 11 (bullseye). We used Docker Engine version 20.10.5. The \setchain algorithms are implemented on CometBFT v0.38, with each ledger server running in a docker container and each container running in a separate machine in the cluster. The containers have no limit on CPU or RAM usage. Each docker container contains one client, one collector module, and one CometBFT server.

\textbf{Experiment Scenarios.} \label{subsec:experimentscenarios}
The parameters considered for the experiments in this work are listed in Table \ref{tab:setchain_param}. The sending rate in Table~\ref{tab:setchain_param} is the total rate at which elements are added across all clients. Each client adds at a rate $sending\_rate/server\_count$ by sending the elements to their local server (the server running in the same docker). The $network\_delay$ parameter is an artificial latency that is added to all communications between servers to simulate the impact of moving from a cluster to a wide area deployment.
The mempool is an important element of CometBFT, where the unconfirmed ledger transactions are held after validation and before being included in the Blockchain by CometBFT. The default mempool setting of CometBFT allows only a maximum of $5,000$ transactions to be present in the mempool. For the \setchain evaluation we did not want this to be a bottleneck. Hence, after a few trial and error experiments, the mempool size has been set to $10,000,000$ transactions or $2$ GB, whichever is reached first.
By default, each experiment has been designed so that clients add elements in the \setchain for 50 seconds. When all the elements have been included in epochs and all the epoch-proofs have been inserted in ledger blocks, the experiment ends. The logs are then collected and analyzed. The injection and throughput metrics are based on the \setchain elements sent by the clients, measured in \textit{elements per second (el/s)}.
To keep the experiments realistic we use as \setchain elements transactions downloaded from Arbitrum \cite{arbitrum}. Regarding the basic algorithms used, we use SHA512 \cite{SHA512} for hashing; signing is done with ed25519, which is part of the EdDSA (Edwards-Curve Digital Signature Algorithm) family \cite{RFC8032,Bernstein2012}, and we use Brotli for compression \cite{rfc7932}. The average size of an Arbitrum transaction is approximately $438$ bytes with a standard deviation of $753.5$. The length of an epoch-proof is $139$ bytes. In Compresschain, the average length of a compressed batch is approximately $16,000$ bytes with a standard deviation of $2,100$ for a collector limit of $100$, and an average of $66,000$ bytes with a standard deviation of $15,000$ for a collector limit of $500$. So, the compression ratio varies approximately from $2.5$ to $3.5$. In Hashchain, the length of a hash-batch is $139$ bytes. In CometBFT, one ledger block is generated roughly every $1.25$ seconds (i.e., the block rate is approximately $0.8$ blocks/s). Unless otherwise stated, the ledger blocksize used in CometBFT is $0.5$ MB.


\subsection{Performance Results}\label{sec:result}

\textbf{Analysis.} We derived an analytical formula to estimate the theoretical highest throughput of each algorithm as a function of system parameters.
These formulas are used in this section to compute the theoretical throughput values, which are then compared with our experimental results.
Due to space constraints, the full derivation is deferred to Appendix~\ref{sec:analysis}.

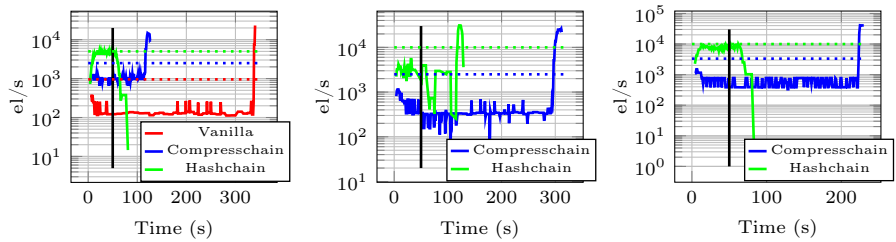
\begin{figure}[t!]
\centering



   

        
    \begin{tikzpicture}
        \begin{axis}[
            xlabel={\scriptsize Time (s)},
            ylabel={\scriptsize el/s},
            ylabel style={yshift=-5pt},
            ymode = log,
            tick label style={font=\scriptsize},
            scale=0.40,
            legend style= {font = \tiny, at={(0.7,0.18)},legend image post style={xscale=0.4},anchor=center, every axis legend/.append style={line width=0.5pt},inner sep=1pt, nodes={inner sep=1pt}},
            ytick={1,10,100,1000,10000,100000},
            grid=both
        ]

        \addplot[line width=1, red] table[col sep=comma, x expr=\coordindex, y index=0]{sr5000cl100/rolling_vanilla.csv};
        \addlegendentry{Vanilla};

        \addplot[line width=1, blue] table[col sep=comma, x expr=\coordindex, y index=0]{sr5000cl100/rolling_compress.csv};
        \addlegendentry{Compresschain};

        \addplot[line width=1, green] table[col sep=comma, x expr=\coordindex, y index=0]{sr5000cl100/rolling_hash.csv};
        \addlegendentry{Hashchain};

        \addplot[color=black, line width=1] coordinates {(50,5) (50,20000)};

        \addplot[color=red, line width=1, dotted] coordinates {(0,955) (350,955)};
        
        \addplot[color=blue, line width=1, dotted] coordinates {(0,2497) (350,2497)};
        
        \addplot[color=green, line width=1, dotted] coordinates {(0,5000) (350,5000)};
    
        \end{axis}
    \end{tikzpicture}
    \label{fig:5000txscl100}
    \begin{tikzpicture}
        \begin{axis}[
            xlabel={\scriptsize Time (s)},
            ylabel={\scriptsize el/s},
            ylabel style={yshift=-5pt},
            ymode = log,
            tick label style={font=\scriptsize},
            scale=0.40,
            legend style= {font = \tiny, at={(0.7,0.13)},legend image post style={xscale=0.4},anchor=center, every axis legend/.append style={line width=0.5pt},inner sep=1pt, nodes={inner sep=1pt}},
            grid=both
        ]
        
        \addplot[line width=1, blue] table[col sep=comma, x expr=\coordindex, y index=0]{sr10000cl100/rolling_compress.csv};
        \addlegendentry{Compresschain};
        
        \addplot[line width=1, green] table[col sep=comma, x expr=\coordindex, y index=0]{sr10000cl100/rolling_hash.csv};
        \addlegendentry{Hashchain};
   
        \addplot[color=black, line width=1] coordinates {(50,20) (50,30000)};
            
        \addplot[color=blue, line width=1, dotted] coordinates {(0,2497) (320,2497)};
        
        \addplot[color=green, line width=1, dotted] coordinates {(0,10000) (320,10000)};
    
        \end{axis}
    \end{tikzpicture}
    \begin{tikzpicture}
        \begin{axis}[
            xlabel={\scriptsize Time (s)},
            ylabel={\scriptsize el/s},
            ylabel style={yshift=-5pt},
            ymode = log,
            tick label style={font=\scriptsize},
            scale=0.40,
            legend style= {font = \tiny, at={(0.7,0.13)},legend image post style={xscale=0.4},anchor=center, every axis legend/.append style={line width=0.5pt},inner sep=1pt, nodes={inner sep=1pt}},
            grid=both,
            ytick={1,10,100,1000,10000,100000}
        ]
        
        \addplot[line width=1, blue] table[col sep=comma, x expr=\coordindex, y index=0]{sr10000cl500/rolling_compress.csv};
        \addlegendentry{Compresschain};
        
        \addplot[line width=1, green] table[col sep=comma, x expr=\coordindex, y index=0]{sr10000cl500/rolling_hash.csv};
        \addlegendentry{Hashchain};
   
        \addplot[color=black, line width=1] coordinates {(50,1) (50,30000)};
        
        \addplot[color=blue, line width=1, dotted] coordinates {(0,3330) (230,3330)};
        
        \addplot[color=green, line width=1, dotted] coordinates {(0,10000) (230,10000)};

        \end{axis}
    \end{tikzpicture}
\caption{Throughput over time of the \setchain algorithms for (left) Sending rate = $5,000$ el/s  and Collector size = $100$, (center) Sending rate = $10,000$ and Collector size = $100$, and (right) Sending rate = $10,000$ el/s and Collector size = $500$.
Solid lines plot the rolling average number of elements committed in 9 seconds. The vertical bar marks the time clients add the last element (roughly after 50 s). The dotted horizontal lines show the minimum of the sending rate and the analytical throughput computed in Section~\ref{sec:results-analysis}.}
\label{fig:throughput_comparison}
\vspace{-1.5em}
\end{figure}
\textbf{Throughput Comparison.}
Fig.~\ref{fig:throughput_comparison} shows the throughput (in elements \textit{committed} per second, el/s) over time of the three \setchain algorithms for different sending rates with $10$ servers and with no increase in the network delay. An element added becomes \textit{committed} when the epoch in which the element has been included gets at least $f+1$ epoch proofs in the ledger. 
The throughput values obtained in the analysis (Section~\ref{sec:results-analysis}) are also shown for reference.
%
%
In Fig.~\ref{fig:throughput_comparison} left, for a sending rate of $5,000$ el/s and collector size $100$, it can be observed that both Vanilla and Compresschain take a long time to commit all the elements and show the peak at the end which is a symptom of stress. 
Moreover, both algorithms show a smaller throughput than the analytical bound. 
Hashchain on the other hand is able to cope and finishes shortly after all elements are added. 
In Fig.~\ref{fig:throughput_comparison} center, for a sending rate of $10,000$ el/s and collector size $100$, we have excluded Vanilla to compare Compresschain and Hashchain better. In this plot it can be seen that both Compresschain and Hashchain are stressed, although the former much more than the latter. As shown in Fig.~\ref{fig:throughput_comparison} right, increasing the collector size to $500$ helps to relieve the stress from Hashchain, while it does not help Compresschain much.
In Table~\ref{tab:throughput}, we present the average throughput achieved up to 50s for all three algorithms in every experiment in Fig.~\ref{fig:throughput_comparison}.

%

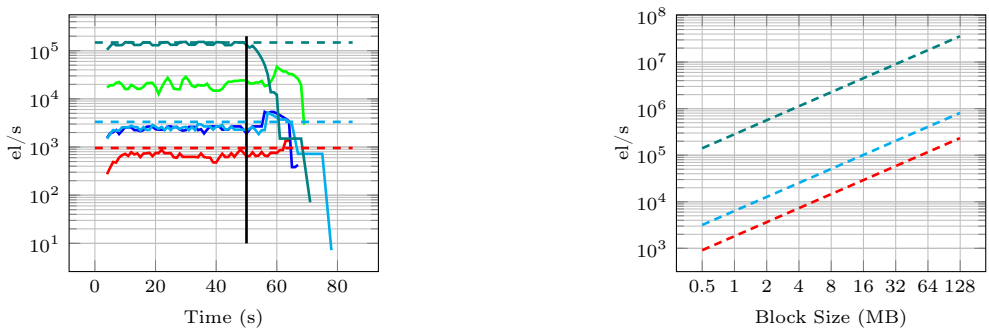
\begin{figure}[h]
\centering
\hspace*{-4em}
\begin{tikzpicture}
        \begin{axis}[
            xlabel={\scriptsize Time (s)},
            ylabel={\scriptsize el/s},
            ylabel style={yshift=-5pt},
            ymode = log,
            tick label style={font=\scriptsize},
            scale=0.6,
            legend to name=highest_throughput_legend, 
            legend columns=1,
            legend cell align = {left},%
            grid=both,
            ytick={10,100,1000,10000,100000,1000000}
        ]

        \addplot[line width=1, red] table[col sep=comma, x expr=\coordindex, y index=0]{sr1000cl100/rolling_vanilla.csv};
        \addlegendentry{Vanilla (1000 el/s)};
        
        \addplot[line width=1, blue] table[col sep=comma, x expr=\coordindex, y index=0]{sr3500cl500w_validation/rolling_compress.csv};
        \addlegendentry{Compresschain (3,500 el/s)};

        \addplot[line width=1, cyan] table[col sep=comma, x expr=\coordindex, y index=0]{sr3500cl500wo_validation/rolling_compress.csv};
        \addlegendentry{Compresschain w/o};
        \addlegendimage{empty legend}\addlegendentry{Decompression (3,500 el/s)} 
        
        \addplot[line width=1, green] table[col sep=comma, x expr=\coordindex, y index=0]{sr30000cl500/rolling_hash.csv};
        \addlegendentry{Hashchain (25,000 el/s)};

        \addplot[line width=1, teal] table[col sep=comma, x expr=\coordindex, y index=0]{sr150000cl500/rolling_hash.csv};
        \addlegendentry{Hashchain w/o Hash-};
        \addlegendimage{empty legend}\addlegendentry{Reversal (150,000 el/s)} 
   
        \addplot[color=black, line width=1] coordinates {(50,10) (50,200000)};

        \addplot[color=red, line width=1, dashed] coordinates {(0,955) (85,955)};
        
        \addplot[color=cyan, line width=1, dashed] coordinates {(0,3330) (85,3330)};
        
        \addplot[color=teal, line width=1, dashed] coordinates {(0,147857) (85,147857)};

        \end{axis}
\end{tikzpicture}
\hspace*{-5em}
\begin{tikzpicture}
   \raisebox{2.5em}{ \scalebox{0.6}{\ref{highest_throughput_legend}}}     
\end{tikzpicture}
\hspace*{-2em}
\begin{tikzpicture}
\begin{axis}[
    domain=5e5:128e6, samples=400,
    xlabel={\scriptsize Block Size (MB)},
    ylabel={\scriptsize el/s},
    ylabel style={yshift=-5pt},
    xtick={5e5,1e6,2e6,4e6,8e6,16e6,32e6,64e6,128e6},
    xticklabels={0.5,1,2,4,8,16,32,64,128}, 
    scaled x ticks=false,
    xminorgrids=true,
    yminorgrids=true,
    ymode = log,
    xmode = log,
    tick label style={font=\scriptsize},
    scale=0.6,
    grid=both,
    ytick={1000,10000,100000,1000000,10000000,100000000}]
  \addplot[line width=1, red, densely dashed]{(0.8*(x-1390))/438};
  \addplot[line width=1, cyan, densely dashed]{x*0.006351558};
  \addplot[line width=1, teal, densely dashed]{0.2820*x};
\end{axis}

\end{tikzpicture}

    
\caption{(left) Highest throughput measured in the evaluation platform with collector size $500$. 
The vertical bar marks the time clients add the last element (roughly after 50 s). Solid lines plot the rolling average number of elements committed in 9 s.  The dashed lines show the analytical throughput. (right) Analytical throughput of the \setchain algorithms for various block sizes with collector size $500$.}
\label{fig:throughput_highest}
\vspace{-1.5em}
\end{figure}

\textbf{Pushing the Hashchain Limits.} \label{subsubsec:pushing_hashchain_limits}
As seen in Fig.~\ref{fig:throughput_comparison} right, there is a noticeable gap between the analytical achievable throughput ($147,857$ el/s, see Section~\ref{sec:results-analysis}) and the actual throughput observed with the implementation of Hashchain, because the experiments did not try large sending rates. To determine the highest achievable throughput, we increased the sending rate. However, we identified a bottleneck around $20,000$ el/s, regardless of further increase in the sending rate beyond that threshold and of the collector size (as seen in Fig.~\ref{fig:throughput_highest} left). The most likely cause of this limitation is the hash-reversal process, where batches are exchanged between servers across the network for each batch-hash sent by the collector. 
%
In our implementation, Setchain servers handle the distribution of transaction batches. More efficient methods could be employed, such as having only a set of $2f+1$ servers sign each batch-hash and epoch, utilizing optimistic validation of hash-batches (share a batch only on request), or implementing alternative distributed batch-sharing mechanisms. To assess the impact of hash-reversal, we conducted experiments removing this service and the validation of hash-batches, while assuming that all servers are correct---meaning all hash-batches are inherently valid.
Fig.~\ref{fig:throughput_highest} left illustrates the highest achieved throughput by Hashchain with and without hash-reversal. From this, it is clear that hash-reversal is the bottleneck, and with a more efficient implementation of hash-reversal, Hashchain would scale much better in throughput. The results confirm that hash-reversal significantly limits performance. Without it, Hashchain reaches an average throughput of $133,882$ el/s for the first $50$ seconds with a sending rate of $150,000$ el/s, compared to an average throughput of $20,061$ el/s for the first $50$ seconds with a sending rate of $25,000$ el/s with hash-reversal enabled. It is also important to note that these results were obtained with a collector size of $500$, chosen to maintain comparability with other experiments in this work. A larger collector size would likely yield even higher throughput without hash-reversal.
For comparison, Fig.~\ref{fig:throughput_highest} (left) also shows the highest achieved throughput by Compresschain and Vanilla. For Compresschain we run it with and without decompression and validation to measure the impact of these processes. Observed throughputs are well below Hashchain’s throughput, even with hash-reversal.
The highest throughputs observed with Vanilla, Compresschain Light, and Hashchain Light are very close to the analytical values.
Fig.~\ref{fig:throughput_highest} (right) shows the analytical throughput for all three \setchain algorithms for larger block sizes, while keeping the other parameters constant (For~\ref{alg:setchain-compress} and~\ref{alg:setchain-hash},  Collector size = $500$). As can be seen, with the usual 4MB blocksize of CommetBFT, Haschain reaches a throughput of $10^6$ el/s, and with blocks of 128 MB reaches more than 30 million el/s.


\begin{figure}[h]
\hspace*{-4em}
\begin{subfigure}[b]{0.45\textwidth}
\centering
\includegraphics[trim={1cm 5cm 10cm 1cm},clip,width=1.3\textwidth]{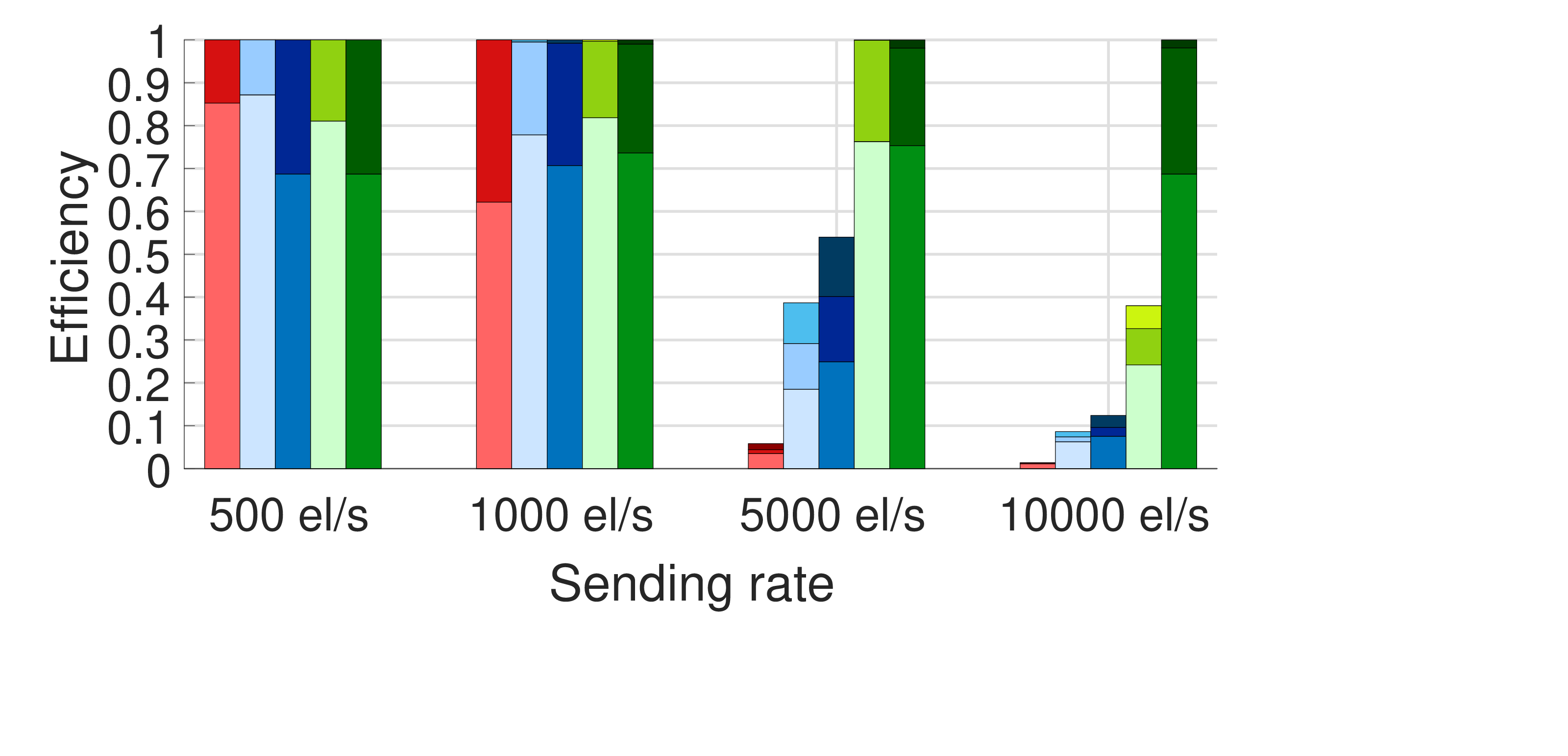}
\caption{\scriptsize Impact of sending rate. 
}
\label{fig:impactofsendrate}
\end{subfigure}
\hfill
\begin{subfigure}[b]{0.45\textwidth}
\centering
\includegraphics[trim={1cm 3cm 6cm 1cm},clip,width=1.3\textwidth]{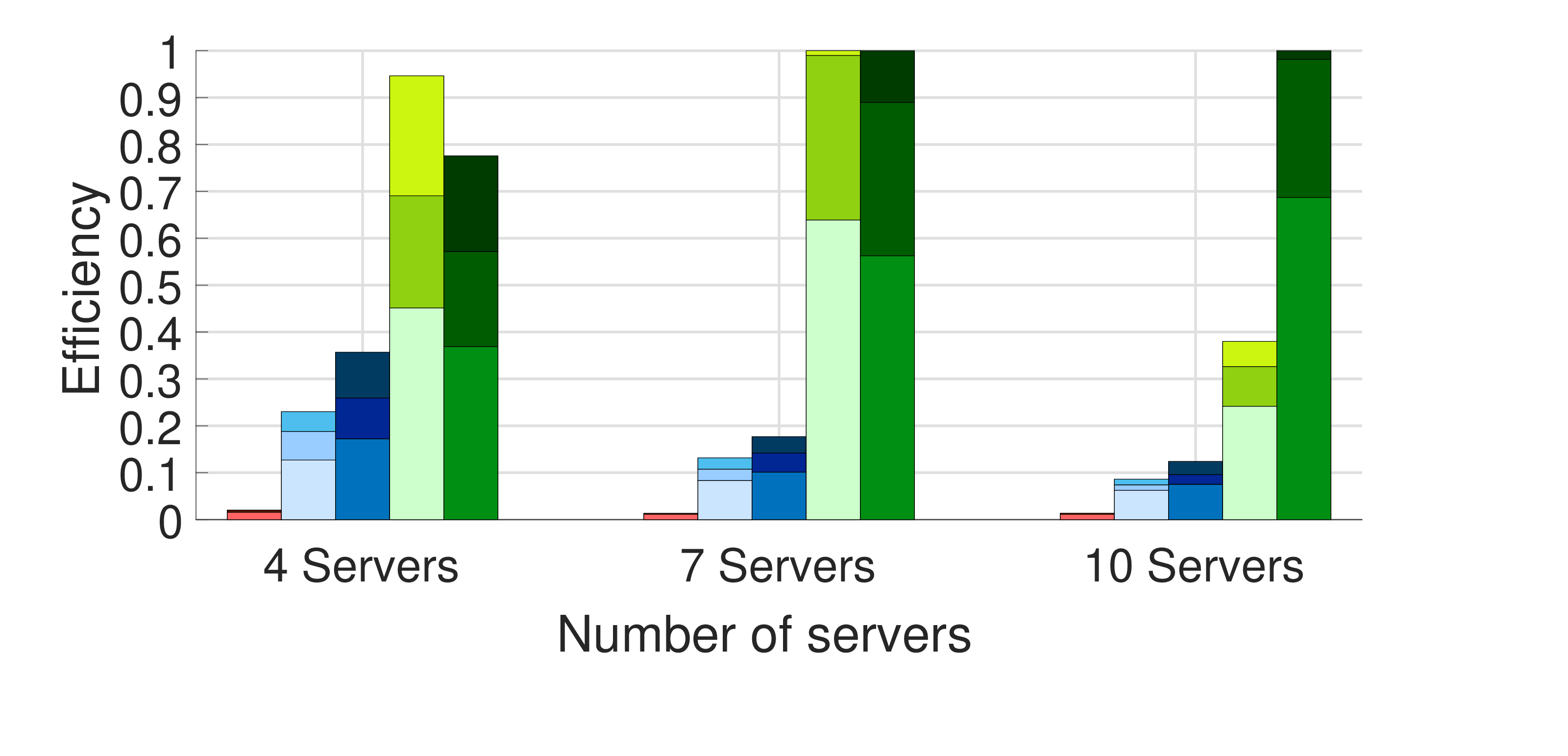}
\caption{\scriptsize Impact of number of servers. 
}
\label{subfig:impactofnofserverssr10000}
\end{subfigure}
\hspace*{-4em}
\begin{subfigure}[b]{0.45\textwidth}
\centering
\includegraphics[trim={1.5cm 3cm 4cm 1cm},clip,width=1.3\textwidth]{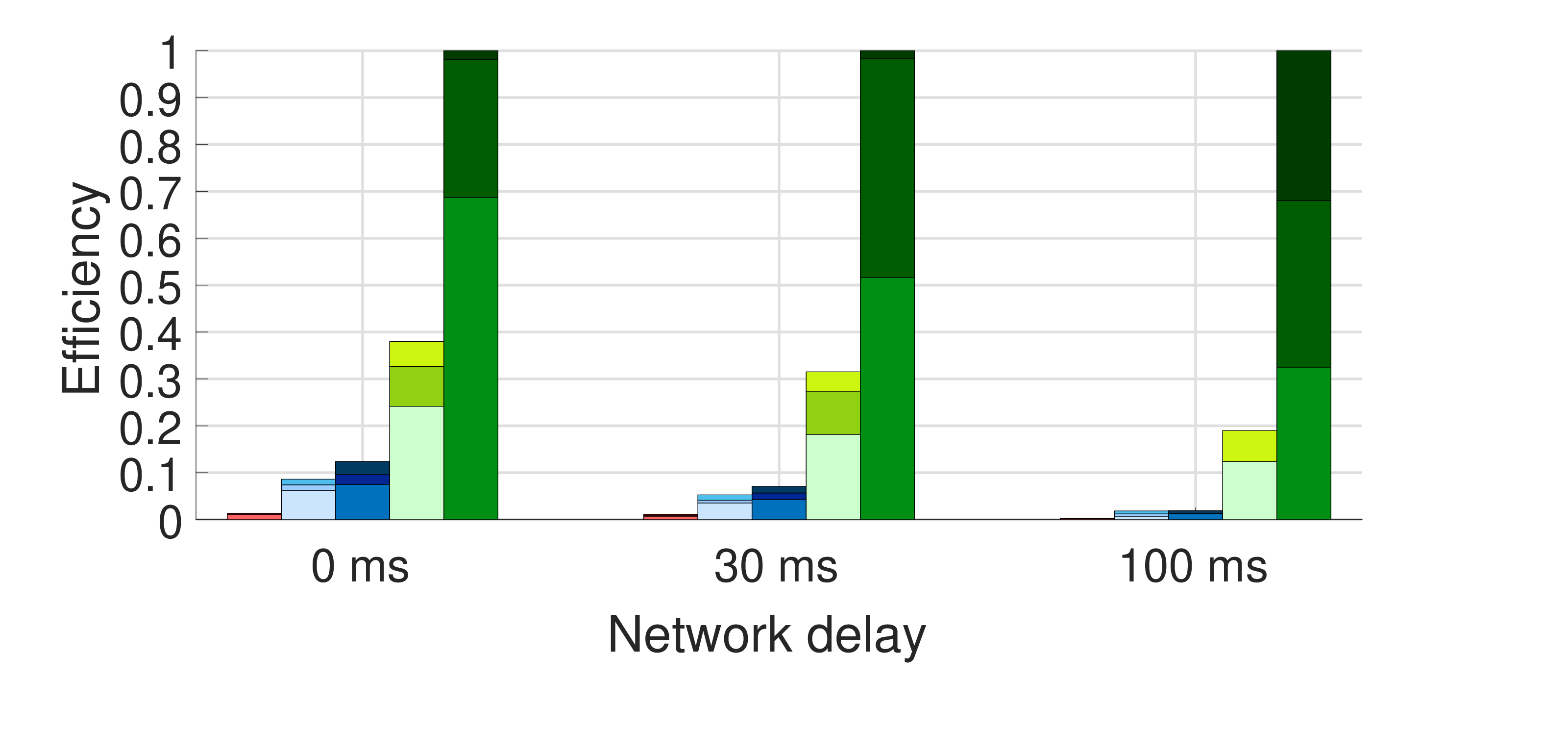}
\caption{\scriptsize Impact of network delay.
}
\label{subfig:impactofnetworkdelaysr10000}
\end{subfigure}
\hspace{8em}
\begin{subfigure}[b]{0.45\textwidth}
\centering
\includegraphics[trim={6cm 1cm 6cm 1cm},clip,width=\textwidth]{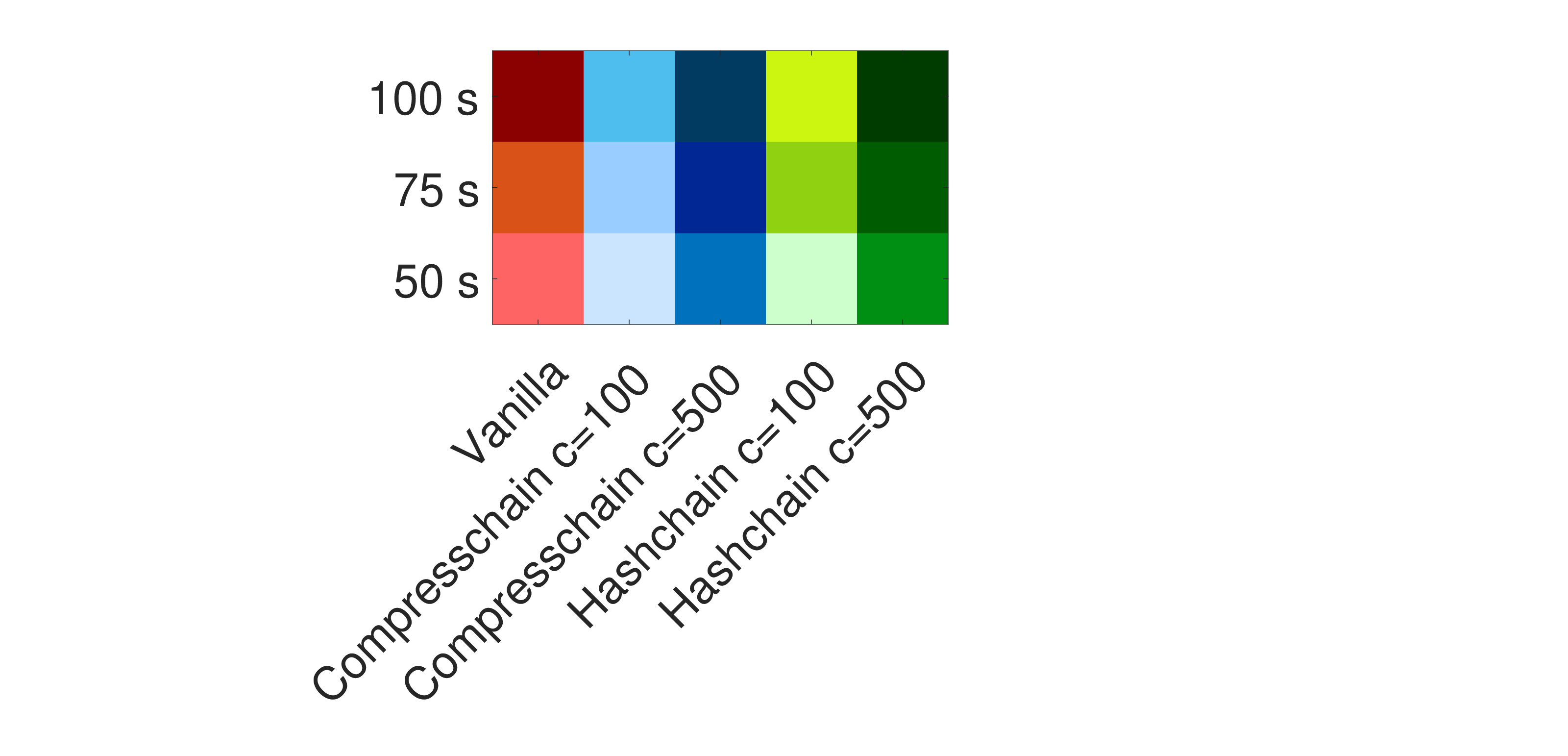}
\caption{\scriptsize Color codes}
\label{subfig:legendefficiencyfactorlarge}
\end{subfigure}
\caption{Efficiency values observed under different scenarios. The base scenario has $10$ servers, a sending rate of $10,000$ el/s, and no ($0$) network delay.}
\label{fig:efficiency}
\vspace{-2.5em}
\end{figure}

\textbf{Efficiency.} \label{subsec:efficiencyfactor}
To quantify easily the level of stress of an algorithm we define and use a metric that we call \emph{efficiency}. The efficiency is obtained by dividing the number of elements committed by the total number of elements added. We calculate the efficiency factor after $50$, $75$, and $100$ seconds. Remember that clients add elements for $50$ seconds in every experiment. If the algorithm is not stressed we expect to observe efficiency close to $1$ after $50$ seconds, and exactly $1$ after $75$ seconds. In this section, we use a base scenario with $10$ servers, a sending rate of $10,000$ el/s, and no ($0$) network delay, and vary one of these parameters at a time.
Fig.~\ref{fig:impactofsendrate} shows the efficiency for four different sending rates ($500$, $1000$, $5,000$, and $10,000$), and for collector sizes $100$ and $500$, when applicable. This figure shows that all algorithms reach full efficiency in 70 seconds for the lower rates $500$ and $1,000$. Then, for rates $5,000$ and $10,000$, Vanilla has very low efficiency. Compresschain, on the other hand, also reduces its efficiency significantly, and increasing the collector size from $100$ to $500$ does not help much. Finally, Hashchain only shows a decrease in efficiency with a rate of $10,000$, which is alleviated by increasing the collector size.
%
%
In Fig.~\ref{subfig:impactofnofserverssr10000},
we compare the change in efficiency observed when the number of servers is varied, maintaining a sending rate of $10,000$ el/s. Observe that Vanilla has the lowest efficiency, even for 4 servers. The efficiency of Compresschain is also low and increasing the collector size does not improve it much. Moreover, it decreases as the number of servers increases.  
Finally, Hashchain only shows low efficiency with 10 servers and a collector size of 100. Interestingly, it shows a less-than-perfect efficiency for 4 servers, which may be due to having fewer servers available for the reverse hashing process.
%
%
%
Fig.~\ref{subfig:impactofnetworkdelaysr10000}
shows how adding an artificial delay to all communication messages affects efficiency. This models the impact of moving from a local-area network to a wide-area network. As can be observed, the increase in network delay reduces the efficiency. However, even with the largest delay of $100$ ms, Hashchain with a collector size of $500$ achieves full efficiency in $100$ seconds.
%

%

\begin{figure}[t!]
\includegraphics[trim={2.5cm 7cm 12cm 0.5cm},clip,width=0.33\textwidth]{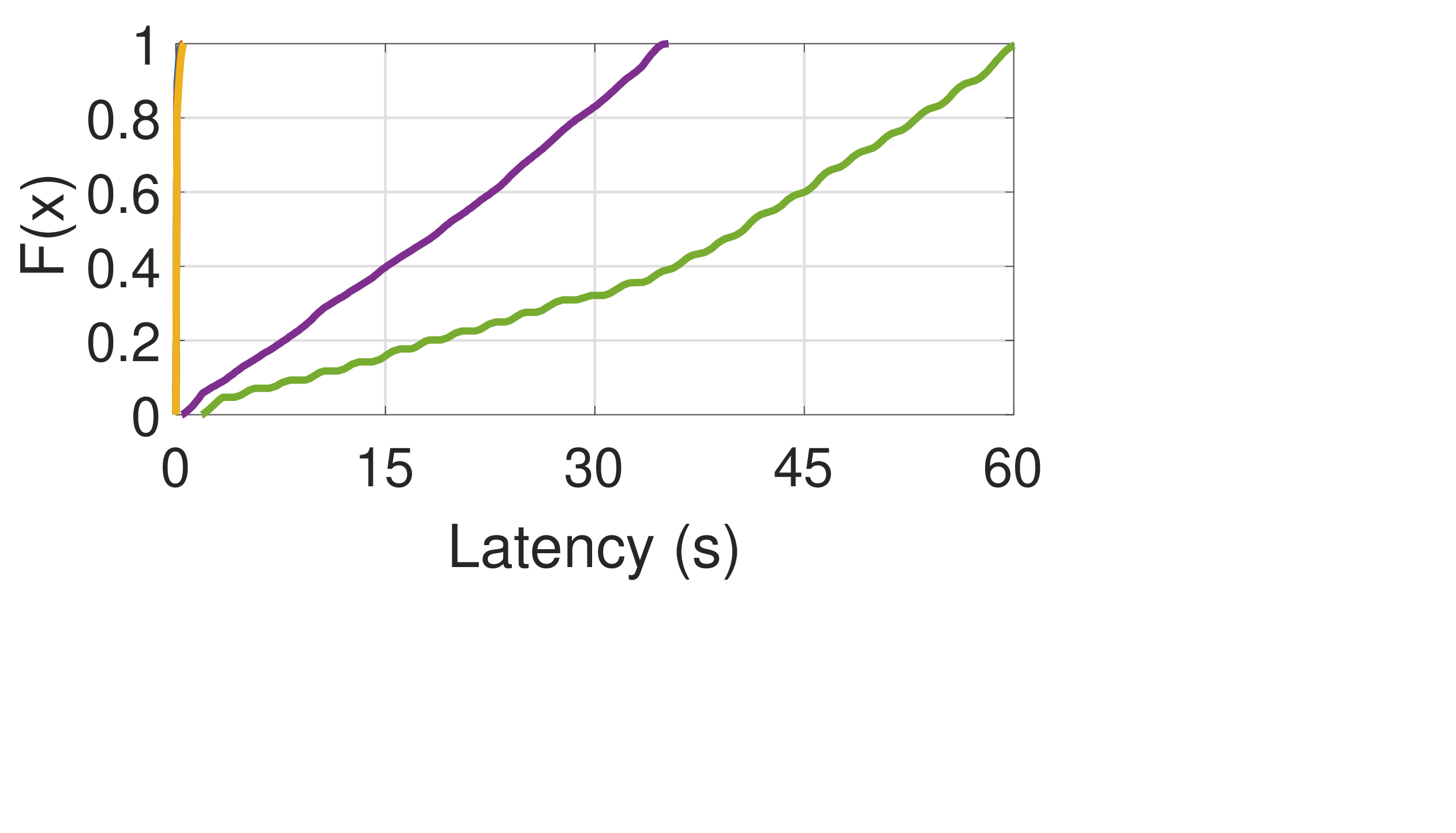}
\label{subfig:latencyvanilla}
\includegraphics[trim={2.5cm 6cm 12.5cm 0.5cm},clip,width=0.31\textwidth]{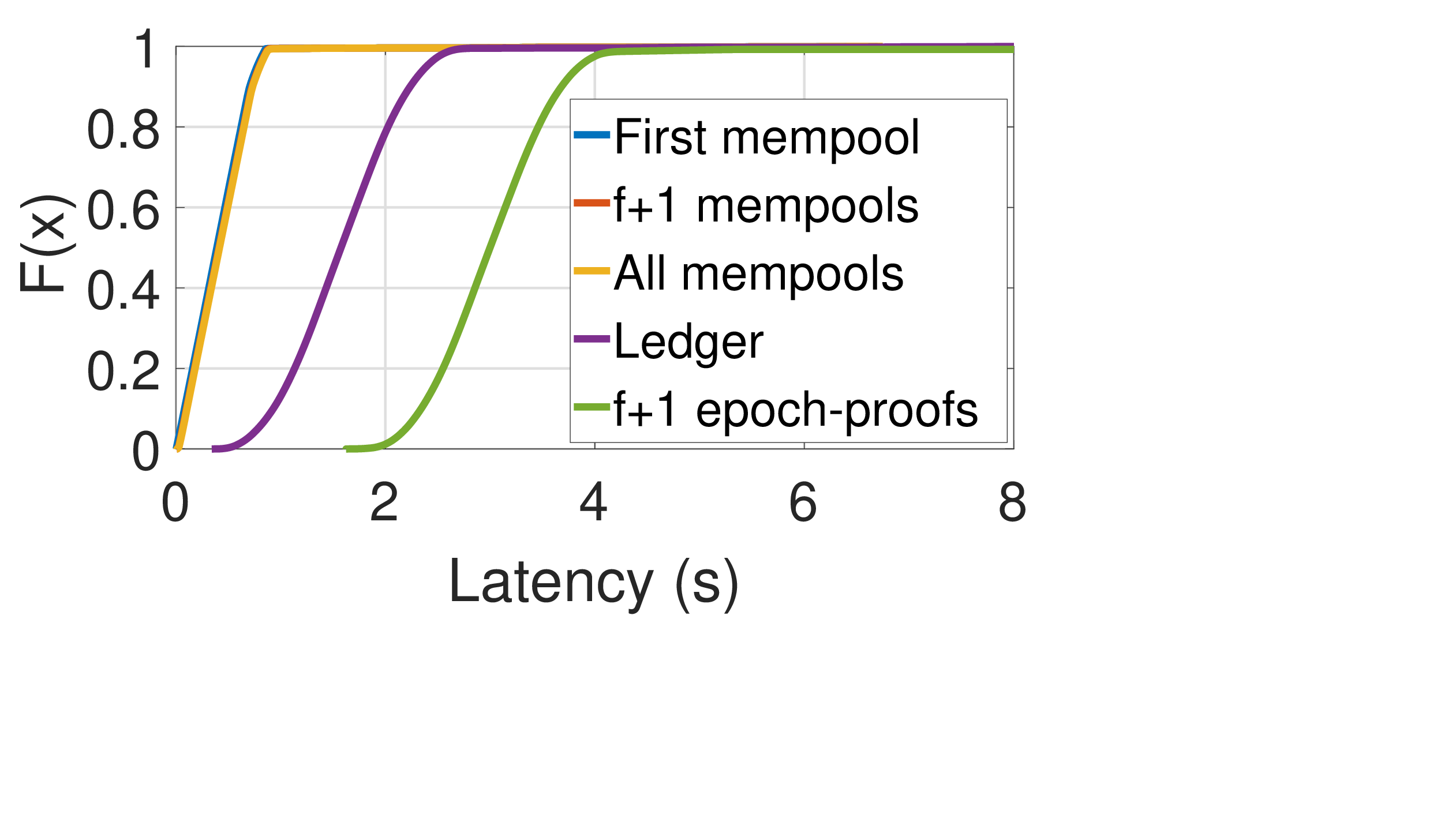}
\label{subfig:latencycompress}
\includegraphics[trim={2.5cm 7cm 12.5cm 0.5cm},clip,width=0.33\textwidth]{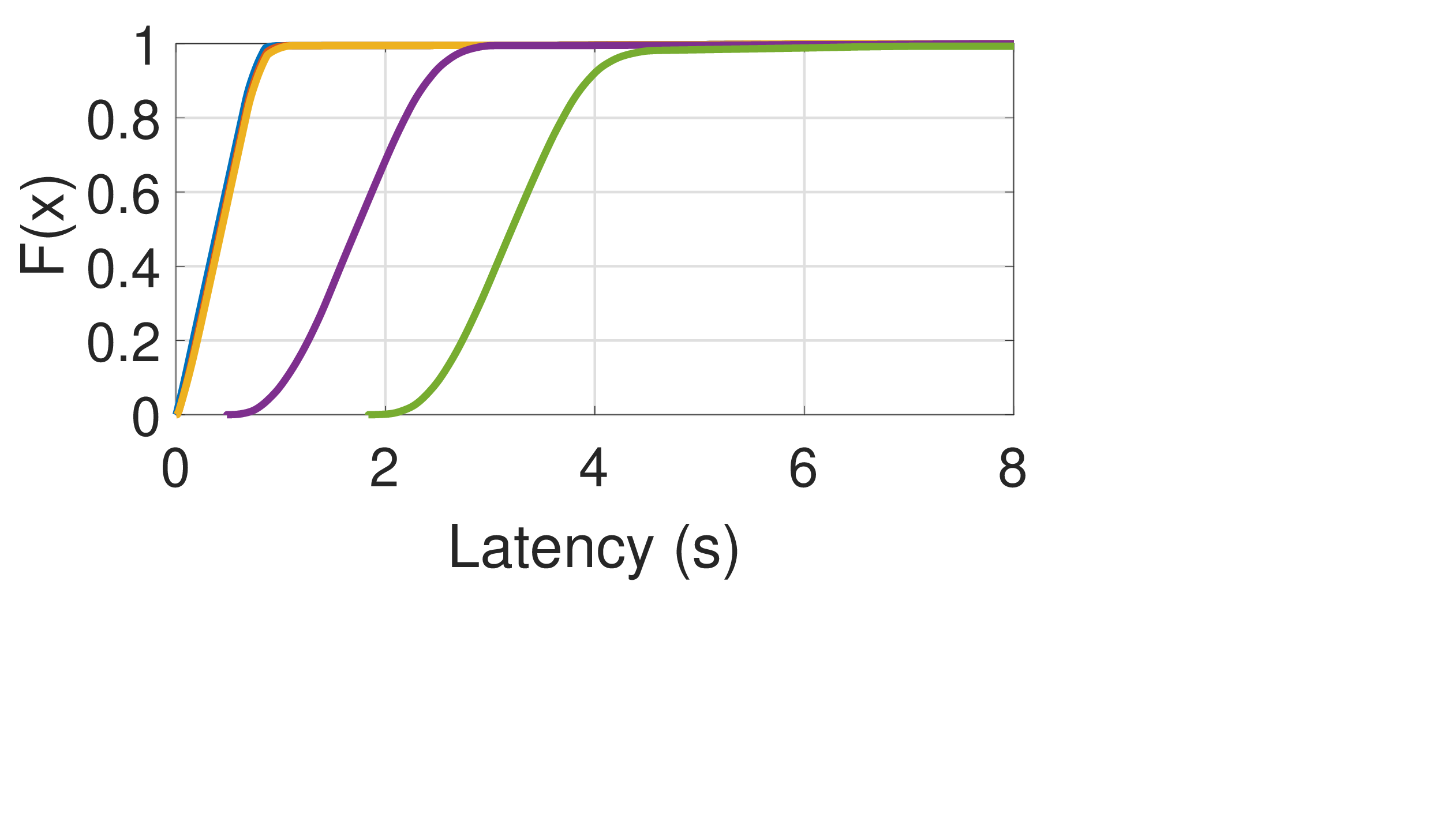}
\label{subfig:latencyhash}
\caption{Cumulative distribution function $F(x)$ of the latency experienced by the elements added to the \setchain to reach several stages in their process for (left) \ref{alg:setchain-vanilla}, (center) \ref{alg:setchain-compress} and (right) \ref{alg:setchain-hash}. Collector size is 100.
The scenario is with $10$ servers, a sending rate of $1,250$ el/s, and no network delay.}
\label{fig:latencygraphs}
\vspace{-1.5em}
\end{figure}

\textbf{Latency.}
When a valid element is sent by a client to its \setchain server, it is eventually sent in a ledger transaction to a CometBFT server.
The transaction, once validated, enters the mempool of the CometBFT server. Then, it is shared with other CometBFT servers via a gossip protocol. These servers validate the transaction and replicate it in their local mempool. Eventually, the transaction will be included in a ledger block, and the \setchain servers will send epoch-proofs to the ledger. When $f+1$ such proofs have been included in ledger blocks, the element has been committed.
Fig.~\ref{fig:latencygraphs} shows the distribution of the latencies experienced by the elements added to the \setchain to reach five different stages until commit, for each algorithm. 
It shows the latency until an element reaches (1) a CometBFT server mempool, (2) $f+1$ CometBFT server mempools, (3) all the CometBFT server mempools, and (4) the ledger. Finally, it shows the (5) commit latency (when $f+1$ epoch-proofs of the epoch in which the element was included are in the ledger).
%
%
In Fig.~\ref{fig:latencygraphs} left, the mempools are reached almost immediately because in Vanilla the elements are directly sent to the CometBFT server. Compresschain and Hashchain, on the other hand, wait for the collector to be full or some timeout before sending the batch of elements to the CometBFT server as a transaction. This explains the delay in mempool latencies in Fig.~\ref{fig:latencygraphs} center and right.
Consider now the latency it takes for an element to reach the ledger and to be committed (with $f+1$ epoch-proofs). With Vanilla these two time gaps are tens of seconds in most cases, while for Compresschain and Hashchain they are usually between one or two seconds. These two algorithms show in this scenario commit latency and finality below 4 seconds with probability almost 1.




\section{Conclusion and Discussions}\label{sec:conclusion}
In this work, we presented three real-world \setchain~implementation
algorithms built on top of a block-based ledger. We have formally
verified that the properties of \setchain~hold for these
algorithms (proofs are sketched for space limitation).
We have implemented the three algorithms and performed an
empirical evaluation to evaluate the effectiveness of each approach in
various scenarios.
Among the three, the Hashchain algorithm consistently performed better
across all evaluated metrics. Its efficient use of hashing techniques
allowed for improved scalability, making it particularly suitable for
large-scale applications. Compresschain, while offering some
improvements over Vanilla, did not match the performance gains
observed with Hashchain.



%
%
%
\bibliographystyle{splncs04}
\bibliography{bibfile}

\begin{thebibliography}{10}
\providecommand{\url}[1]{\texttt{#1}}
\providecommand{\urlprefix}{URL }
\providecommand{\doi}[1]{https://doi.org/#1}

\bibitem{danksharding}
Ethereum roadmap - danksharding. \url{https://ethereum.org/en/roadmap/danksharding/}

\bibitem{zksync-era}
Matter labs - zksync era. \url{https://github.com/matter-labs/zksync-era}

\bibitem{ethereumblocksize}
Ethereum development documentation - block size (2025), \url{https://ethereum.org/en/developers/docs/blocks/#block-size}

\bibitem{rfc7932}
Alakuijala, J., Szabadka, Z.: {Brotli Compressed Data Format}. RFC 7932 (Jul 2016). \doi{10.17487/RFC7932}, \url{https://www.rfc-editor.org/info/rfc7932}

\bibitem{barger2021byzantinefaulttolerantconsensuslibrary}
Barger, A., Manevich, Y., Meir, H., Tock, Y.: A byzantine fault-tolerant consensus library for hyperledger fabric (2021), \url{https://arxiv.org/abs/2107.06922}

\bibitem{Bernstein2012}
Bernstein, D.J., Duif, N., Lange, T., Schwabe, P., Yang, B.Y.: High-speed high-security signatures. Journal of Cryptographic Engineering  \textbf{2}(2),  77--89 (Sep 2012). \doi{10.1007/s13389-012-0027-1}, \url{https://doi.org/10.1007/s13389-012-0027-1}

\bibitem{blockcerts}
Blockcerts: Blockcerts: The open standard for blockchain credentials, \url{https://github.com/blockchain-certificates}

\bibitem{blocknative_ethereum_merge}
Blocknative: Ethereum merge: The impact of proof-of-stake (2022), \url{https://www.blocknative.com/blog/ethereum-merge-proof-of-stake}

\bibitem{fantastyc}
Boitier, W., Del~Pozzo, A., Garcia-Perez, A., Gazut, S., Jobic, P., Lemaire, A., Mahe, E., Mayoue, A., Perion, M., Rezende, T.F., Singh, D., Tucci-Piergiovanni, S.: Fantastyc: Blockchain-based federated learning made secure and practical. In: 2024 43rd International Symposium on Reliable Distributed Systems (SRDS). pp. 260--270 (2024). \doi{10.1109/SRDS64841.2024.00033}

\bibitem{Buchman.2018.Tendermint}
Buchman, E., Kwon, J., Milosevic, Z.: The latest gossip on {BFT} consensus. CoRR  \textbf{abs/1807.04938} (2018), \url{http://arxiv.org/abs/1807.04938}

\bibitem{buchnik2019fireledgerhighthroughputblockchain}
Buchnik, Y., Friedman, R.: Fireledger: A high throughput blockchain consensus protocol (2019), \url{https://arxiv.org/abs/1901.03279}

\bibitem{capretto2024improving}
Capretto, M., Ceresa, M., {Fern{\'a}ndez Anta}, A., Russo, A., S\'{a}nchez, C.: Improving blockchain scalability with the setchain data-type. Distributed Ledger Technologies: Research and Practice  \textbf{3}(2) (jun 2024). \doi{10.1145/3626963}, \url{https://doi.org/10.1145/3626963}

\bibitem{tendermint.design}
Cason, D., Fynn, E., Milosevic, N., Milosevic, Z., Buchman, E., Pedone, F.: The design, architecture and performance of the tendermint blockchain network. In: 2021 40th International Symposium on Reliable Distributed Systems (SRDS). pp. 23--33 (2021). \doi{10.1109/SRDS53918.2021.00012}

\bibitem{DBLP:journals/tocs/CastroL02}
Castro, M., Liskov, B.: Practical byzantine fault tolerance and proactive recovery. {ACM} Trans. Comput. Syst.  \textbf{20}(4),  398--461 (2002). \doi{10.1145/571637.571640}, \url{https://doi.org/10.1145/571637.571640}

\bibitem{redbelly}
Crain, T., Natoli, C., Gramoli, V.: Red belly: A secure, fair and scalable open blockchain. In: 2021 IEEE Symposium on Security and Privacy (SP). pp. 466--483 (2021). \doi{10.1109/SP40001.2021.00087}

\bibitem{narwhal}
Danezis, G., Kokoris-Kogias, L., Sonnino, A., Spiegelman, A.: Narwhal and tusk: a dag-based mempool and efficient bft consensus. In: Proceedings of the Seventeenth European Conference on Computer Systems. p. 34–50. EuroSys '22, Association for Computing Machinery, New York, NY, USA (2022). \doi{10.1145/3492321.3519594}, \url{https://doi.org/10.1145/3492321.3519594}

\bibitem{hyperledger_fabric_v3_0_0}
Fabric, H.: Hyperledger fabric v3.0.0 release (2023), \url{https://github.com/hyperledger/fabric/releases/tag/v3.0.0}

\bibitem{anta2021principles}
{Fern{\'a}ndez Anta}, A., Georgiou, C., Herlihy, M., Potop-Butucaru, M.: Principles of Blockchain Systems. Morgan \& Claypool Publishers (2021)

\bibitem{anta2018formalizing}
{Fern{\'a}ndez Anta}, A., Konwar, K., Georgiou, C., Nicolaou, N.: Formalizing and implementing distributed ledger objects. ACM Sigact News  \textbf{49}(2),  58--76 (2018)

\bibitem{fu2024quantifying}
Fu, Y., Jing, M., Zhou, J., Wu, P., Wang, Y., Zhang, L., Hu, C.: Quantifying the blockchain trilemma: A comparative analysis of algorand, ethereum 2.0, and beyond (2024), \url{https://ar5iv.labs.arxiv.org/html/2407.14335v1}, accessed: 2024-10-01

\bibitem{gilad2017algorand}
Gilad, Y., Hemo, R., Micali, S., Vlachos, G., Zeldovich, N.: Algorand: Scaling byzantine agreements for cryptocurrencies. In: Proceedings of the 26th Symposium on Operating Systems Principles. p. 51–68. SOSP '17, Association for Computing Machinery, New York, NY, USA (2017). \doi{10.1145/3132747.3132757}, \url{https://doi.org/10.1145/3132747.3132757}

\bibitem{diablo}
Gramoli, V., Guerraoui, R., Lebedev, A., Natoli, C., Voron, G.: Diablo: A benchmark suite for blockchains. In: Proceedings of the Eighteenth European Conference on Computer Systems. p. 540–556. EuroSys '23, Association for Computing Machinery, New York, NY, USA (2023). \doi{10.1145/3552326.3567482}, \url{https://doi.org/10.1145/3552326.3567482}

\bibitem{dumbo}
Guo, B., Lu, Z., Tang, Q., Xu, J., Zhang, Z.: Dumbo: Faster asynchronous bft protocols. In: Proceedings of the 2020 ACM SIGSAC Conference on Computer and Communications Security. p. 803–818. CCS '20, Association for Computing Machinery, New York, NY, USA (2020). \doi{10.1145/3372297.3417262}, \url{https://doi.org/10.1145/3372297.3417262}

\bibitem{RFC8032}
Josefsson, S., Liusvaara, I.: Rfc 8032: Edwards-curve digital signature algorithm (eddsa) (2017)

\bibitem{arbitrum}
Kalodner, H., Goldfeder, S., Chen, X., Weinberg, S.M., Felten, E.W.: Arbitrum: Scalable, private smart contracts. In: 27th USENIX Security Symposium (USENIX Security 18). pp. 1353--1370. USENIX Association, Baltimore, MD (Aug 2018), \url{https://www.usenix.org/conference/usenixsecurity18/presentation/kalodner}

\bibitem{omniledger}
Kokoris-Kogias, E., Jovanovic, P., Gasser, L., Gailly, N., Syta, E., Ford, B.: Omniledger: A secure, scale-out, decentralized ledger via sharding. In: 2018 IEEE Symposium on Security and Privacy (SP). pp. 583--598 (2018). \doi{10.1109/SP.2018.000-5}

\bibitem{honeybadger}
Miller, A., Xia, Y., Croman, K., Shi, E., Song, D.: The honey badger of bft protocols. In: Proceedings of the 2016 ACM SIGSAC Conference on Computer and Communications Security. p. 31–42. CCS '16, Association for Computing Machinery, New York, NY, USA (2016). \doi{10.1145/2976749.2978399}, \url{https://doi.org/10.1145/2976749.2978399}

\bibitem{SHA512}
{National Institute of Standards and Technology (NIST)}: {FIPS PUB 180-4: Secure Hash Standard (SHS)}. \url{https://nvlpubs.nist.gov/nistpubs/FIPS/NIST.FIPS.180-4.pdf} (Mar 2012), federal Information Processing Standards Publication

\bibitem{mitblockchaindiploma}
News, M.: Digital diploma debuts at mit (2017), \url{https://news.mit.edu/2017/mit-debuts-secure-digital-diploma-using-bitcoin-blockchain-technology-1017}

\bibitem{plasma}
Poon, J., Buterin, V.: Plasma: Scalable autonomous smart contracts (2017), \url{https://plasma.io/plasma.pdf}

\bibitem{lightning}
Poon, J., Dryja, T.: The bitcoin lightning network: Scalable off-chain instant payments (2016), \url{https://lightning.network/lightning-network-paper.pdf}

\bibitem{zlb}
Ranchal-Pedrosa, A., Gramoli, V.: Zlb: A blockchain to tolerate colluding majorities. In: 2024 54th Annual IEEE/IFIP International Conference on Dependable Systems and Networks (DSN). pp. 209--222 (2024). \doi{10.1109/DSN58291.2024.00032}

\bibitem{ethereumtps}
Rao, I.S., Kiah, M.L.M., Hameed, M.M., Memon, Z.A.: Scalability of blockchain: a comprehensive review and future research direction. Cluster Computing  \textbf{27}(5),  5547--5570 (Aug 2024). \doi{10.1007/s10586-023-04257-7}, \url{https://doi.org/10.1007/s10586-023-04257-7}

\bibitem{DBLP:conf/blockchain2/RussoAVR21}
Russo, A., {Fern{\'{a}}ndez Anta}, A., {Gonz{\'{a}}lez Vasco}, M.I., Romano, S.P.: Chirotonia: {A} scalable and secure e-voting framework based on blockchains and linkable ring signatures. In: Xiang, Y., Wang, Z., Wang, H., Niemi, V. (eds.) 2021 {IEEE} International Conference on Blockchain, Blockchain 2021, Melbourne, Australia, December 6-8, 2021. pp. 417--424. {IEEE} (2021). \doi{10.1109/BLOCKCHAIN53845.2021.00065}, \url{https://doi.org/10.1109/Blockchain53845.2021.00065}

\bibitem{bullshark}
Spiegelman, A., Giridharan, N., Sonnino, A., Kokoris-Kogias, L.: Bullshark: Dag bft protocols made practical. In: Proceedings of the 2022 ACM SIGSAC Conference on Computer and Communications Security. p. 2705–2718. CCS '22, Association for Computing Machinery, New York, NY, USA (2022). \doi{10.1145/3548606.3559361}, \url{https://doi.org/10.1145/3548606.3559361}

\bibitem{malachite}
Systems, I.: Malachite: Decentralize whatever, \url{https://github.com/informalsystems/malachite}

\bibitem{atomicwallet_tps}
Wallet, A.: What is tps (transactions per second)? (2023), \url{https://atomicwallet.io/academy/articles/what-is-tps}

\bibitem{webisoft_cosmos_tps}
Webisoft: Cosmos tps (transactions per second): Understanding the speed of the internet of blockchains (2023), \url{https://webisoft.com/articles/cosmos-tps/}

\bibitem{bitcoin_confirmation_wiki}
Wiki, B.: Confirmation - bitcoin wiki (2023), \url{https://en.bitcoin.it/wiki/Confirmation}, accessed: 2024-10-01

\bibitem{hotstuff}
Yin, M., Malkhi, D., Reiter, M.K., Gueta, G.G., Abraham, I.: Hotstuff: Bft consensus with linearity and responsiveness. In: Proceedings of the 2019 ACM Symposium on Principles of Distributed Computing. p. 347–356. PODC '19, Association for Computing Machinery, New York, NY, USA (2019). \doi{10.1145/3293611.3331591}, \url{https://doi.org/10.1145/3293611.3331591}

\bibitem{rapidchain}
Zamani, M., Movahedi, M., Raykova, M.: Rapidchain: Scaling blockchain via full sharding. In: Proceedings of the 2018 ACM SIGSAC Conference on Computer and Communications Security. p. 931–948. CCS '18, Association for Computing Machinery, New York, NY, USA (2018). \doi{10.1145/3243734.3243853}, \url{https://doi.org/10.1145/3243734.3243853}

\bibitem{p-hotstuff}
Zhu, F., You, L., Wang, J., Li, L.: P-hotstuff: Parallel bft algorithm with throughput insensitive to propagation delay. Computer Networks  \textbf{262},  111183 (2025). \doi{https://doi.org/10.1016/j.comnet.2025.111183}, \url{https://www.sciencedirect.com/science/article/pii/S1389128625001513}

\end{thebibliography}
\newpage

\appendix

\section{Related Works}\label{sec:relatedworks}
Scalability remains a central challenge as blockchains gain
popularity. Various techniques have been proposed to address this
issue.
Bitcoin, one of the most widely known blockchains, operates on Proof-of-Work (PoW), offering robust security and decentralization. However, Bitcoin suffers from low throughput, averaging around $7$ TPS, due to its block size and 10-minute block intervals~\cite{atomicwallet_tps}. Additionally, Bitcoin’s latency is significant, with blocks typically taking $10$ minutes to confirm, and (optimistic) finality is reached after about $60$ minutes with the six-block confirmation rule \cite{bitcoin_confirmation_wiki}.

Ethereum, which now uses PoS, handled $12$ to $30$ TPS~\cite{ethereumtps} with PoW. Finality is generally reached within $12$ to $15$ minutes, as two ``epochs'' (each $32$ slots or $6.4$ minutes) are justified and finalized~\cite{blocknative_ethereum_merge}. Sharding~\cite{omniledger,rapidchain}, is an on-chain solution which partitions the blockchain into smaller segments to enable parallel transaction processing. Though ``The Merge'' did not improve Ethereum's throughput, as its main goal was to transition to PoS, Ethereum’s Danksharding initiative, part of the larger Ethereum 2.0 roadmap, is projected to increase throughput to more than $100,000$ TPS once fully implemented~\cite{danksharding}.

While sharding is still under development, Layer 2 solutions like rollups provide immediate scalability benefits. Layer 2 solutions include state channels~\cite{lightning}, sidechains, and various forms of rollups, such as Plasma~\cite{plasma}, Optimistic Rollups~\cite{arbitrum}, and Zero Knowledge Rollups~\cite{zksync-era}. These technologies facilitate high-volume transaction processing off the main blockchain, thereby alleviating network congestion and enhancing transaction speed and cost efficiency.

Arbitrum~\cite{arbitrum}, a Layer 2 solution for Ethereum, enhances scalability using Optimistic Rollups.  Its recent upgrade, Arbitrum Nitro, theoretically enables the network to achieve up to $40,000$ TPS. However, latency and finality time for transactions on Arbitrum depend on the Ethereum Layer 1 latency and finality time, as Arbitrum posts batches of transactions back to Ethereum for settlement. Arbitrum employs a fraud-proof mechanism, with a one-week challenge period for resolution, though challenges are uncommon.

Algorand~\cite{fu2024quantifying}, using a Pure Proof-of-Stake (PPoS) consensus mechanism, offers a throughput of around $6,000$ TPS. It provides fast finality, with transactions confirmed in approximately $3.5$ seconds. 
Cosmos~\cite{webisoft_cosmos_tps}, which employs the Tendermint consensus algorithm~\cite{Buchman.2018.Tendermint}, can handle around $10,000$ TPS. Tendermint provides finality within $6$-$7$ seconds. However, the performance depends on the specific blockchain’s design, optimizations, and network conditions.

Hyperledger Fabric, a permissioned blockchain, introduced Byzantine Fault Tolerant (BFT) ordering service in version v3.0.0~\cite{hyperledger_fabric_v3_0_0} through the SmartBFT consensus library~\cite{barger2021byzantinefaulttolerantconsensuslibrary}, enabling Fabric to handle Byzantine faults. Fabric typically achieves thousands of TPS depending on the configuration, with fast finality upon block commitment. However, the addition of BFT may slightly increase latency due to added consensus steps.

Another approach to scaling blockchain is consensus optimization. By implementing a Set Byzantine Consensus, Red Belly~\cite{redbelly} enhances the scalability, allowing the network to agree on a superblock of all proposed blocks. RedBelly achieves impressive scalability in testing environments, reaching over $660,000$ TPS, and up to $30,000$ TPS in real-world conditions. Fast finality ensures that transactions are final and irreversible within $3$ seconds, further enhancing its appeal for high-throughput applications. Algorithms for Set Byzantine Consensus (SBC) or that use SBC are an active area of research \cite{redbelly,capretto2024improving,zlb} that exploits the lack of order between transactions. 

Hotstuff~\cite{hotstuff} is a leader-based BFT replication consensus protocol with linear communication complexity and optimistic responsiveness. The Hotstuff paper itself proposes a variant of the protocol called Chained HotStuff, which is a pipelined Basic HotStuff where a Quorum Certificate can serve in different phases simultaneously. There are variants of this work, being proposed to improve its performance. For example, P-Hotstuff~\cite{p-hotstuff}, which introduces parallelism into Hotstuff, claims that it achieves an average throughput which is 20 times that of Hotstuff.

FireLedger~\cite{buchnik2019fireledgerhighthroughputblockchain} is a Byzantine Fault Tolerant (BFT) consensus protocol designed to optimize throughput and latency in permissioned blockchain environments by leveraging the iterative nature of blockchains to improve their throughput in optimistic execution scenarios. In their evaluation, the authors report that FireLedger, when deployed on ten mid-range Amazon instances within a single data center, can reach a throughput of approximately 160K TPS for transactions of size 512 bytes. In a geo-distributed deployment across ten Amazon nodes, the protocol achieves around 30K TPS for the same transaction size.

Dumbo~\cite{dumbo} is an improvement on Honeybadger~\cite{honeybadger} - the first practical asynchronous BFT protocol. Dumbo replaces HoneyBadger’s $n$ concurrent Asynchronous Binary Agreements (ABAs) with a single Multi‑valued Validated Byzantine Agreement (MVBA), improving both latency and throughput. Dumbo-NG~\cite{dumbo} builds on this idea by decoupling transaction dissemination from agreement and executing them concurrently. It uses a bandwidth‑oblivious MVBA with threshold signatures, achieving 4–8× higher peak throughput and latency that stays stable as throughput grows.

Narwhal\cite{narwhal} decouples reliable transaction dissemination from ordering via a DAG-based mempool that certifies availability and supports bounded-memory garbage collection and scale-out workers. When paired with HotStuff (Narwhal-HS), it sustains about 130–140k TPS at less than 2s latency in WAN settings with 50 validators, and adding workers scales throughput roughly linearly to 600k TPS without increasing latency. Tusk adds the actual consensus on top of Narwhal: a fully asynchronous, zero-message-overhead protocol that orders by interpreting the local DAG with a shared random coin. In WAN experiments, it reaches 170k TPS at 3s with 50 validators, and under 1 and 3 crash faults, it keeps latency under 4s and 6s while maintaining high throughput.

Bullshark~\cite{bullshark} is a DAG-based Byzantine atomic broadcast protocol optimized for the common (partially synchronous) case: it adds a fast path that commits in 2 rounds during synchrony while retaining an asynchronous fallback with 6-round expected latency and O(n) amortized communication. Built atop Narwhal\cite{narwhal}, it decouples data dissemination from ordering, avoids view-change or synchronization, and provides timely fairness with bounded memory via garbage collection. In evaluation, a partially synchronous Bullshark variant reaches around 125K TPS at 2s latency with 50 parties.

Fantastyc~\cite{fantastyc} proposes a byzantine-tolerant design that anchors only cryptographic proofs on-chain while keeping client updates and aggregated models off-chain in a fault-tolerant key-value store keyed by the data’s hash. A server collects updates, stores \texttt{(hash(v), v)} off-chain, and, after gathering $f+1$ signed hash-keys, produces a Proof of Availability \& Integrity (PoA\&I). The blockchain records sets of PoA\&I (typically one transaction per round) rather than raw data, and clients later retrieve and verify values directly from the store using the anchored hashes. This cleanly decouples ordering (chain), integrity (servers), and availability (storage) and lets light clients validate with the PoA\&I alone. Conceptually, this mirrors Hashchain’s approach of appending hashes of batches on the ledger and consolidating an epoch only after $f+1$ signatures on the batch hash, to ensure at least one honest replica can serve the data, thereby slashing on-chain bandwidth while preserving retrievability.
\section{Algorithm Vanilla}\label{sec:algorithmvanilla}
    \begin{algorithm}[H]
      \renewcommand{\thealgorithm}{Vanilla}         
      \caption{\small 
      Code executed by server $v$.}%
      \label{alg:setchain-vanilla}%
      \small
      \begin{multicols}{2}
      \begin{algorithmic}[1]
      \State Use block-based ledger $\ledgerobject$ shared by all servers
      \State Init: $\TheSet \leftarrow \emptyset$, $\epoch \leftarrow 0$, 
      \Statex $\history \leftarrow \emptyset$, $\proofs \leftarrow \emptyset$
      \Function{\add}{$e$}
        \State assert $\validelement(e) \land $
        \Statex $e \notin \TheSet$ 
        \label{alg:setchain-vanilla:valid_element}
            \State $\TheSet \leftarrow \TheSet \cup \{ e \}$ \label{alg:setchain-vanilla-setaddition}
            \State $\ledgerobject.\Append(e)$ \label{alg:setchain-vanilla-append}
        \State \textbf{return} 
      \EndFunction
      
        \Function{\get}{\null}
            \State \hspace*{-1em} \textbf{return} $(\TheSet, \history, \epoch, \proofs)$ \label{alg:setchain-vanilla-get}
        \EndFunction

        \Upon{$\ledgerobject.\NewBlock(B)$} \label{alg:vanilla-newblock}
        \State $np \leftarrow \{ep \in B:$
        \Statex $ ep=\langle j, p, w \rangle \text{~is an epoch-proof}\land \validproof(j, p, w, \history[j])\}$
            \State $\proofs \leftarrow \proofs \cup np$ \label{alg:setchain-vanilla:validepochproof}
            \State $G \leftarrow \{e \in B: e \text{~is an element} \land \validelement(e) \land e \notin \history\}$ \label{alg:setchain-vanilla-othervalidelement}
            \State $\TheSet \leftarrow \TheSet \cup G$ \label{alg:setchain-vanilla-othersetaddition}
            \State $\epoch \leftarrow \epoch + 1$ \label{alg:setchain-vanilla-epochinc}
            \State $\history[\epoch] \leftarrow G$  \label{alg:setchain-vanilla-epochaddition}
            \State $p \leftarrow \Sign_v(\Hash(\epoch,G))$
            \State $\ledgerobject.\Append(\langle \epoch, p, v \rangle)$ \label{alg:setchain-vanilla:epochproof}
        \EndUpon
        \end{algorithmic}
        \end{multicols}
      \end{algorithm}
\section{Proof of Correctness} \label{sec:proofofcorrectness}
In this section, we present the proof of correctness for the \setchain algorithms presented in Section~\ref{sec:algorithms}.

\subsection{Correctness of Algorithm~\ref{alg:setchain-vanilla}}

Algorithm~\ref{alg:setchain-vanilla} guarantees Property~\ref{prop:consistent-sets}. 

\begin{lem}\label{lem:vanilla-consistent-sets}
     Let $(T,H,h,P)=\setobject.\get_v()$ be an invocation to a correct server $v$. Then, $\forall i \in \{1,\ldots,h\}, H[i] \subseteq T$.
\end{lem}

\begin{proof}
We prove that in a correct server $v$ it holds that $\forall i \in [1,\epoch], \history[i] \subseteq \TheSet$ at all times.

Initially, we have that $\epoch=0$, since Line \ref{alg:setchain-vanilla-epochinc} was never executed. Hence, Line \ref{alg:setchain-vanilla-epochaddition} was not executed either, $\history$ is empty, and the claim trivially holds. 

Let us assume now that $\epoch>0$ and consider some $i \in [1,\epoch]$. Since $\epoch$ is only modified in Line~\ref{alg:setchain-vanilla-epochinc} by increments of 1, then, when it became equal to $i$, all the elements $G$ that were later added to $\history$ in epoch $i$ (in Line~\ref{alg:setchain-vanilla-epochaddition}) were guaranteed to be in $\TheSet$ before (in Line~\ref{alg:setchain-vanilla-othersetaddition}), and are never removed from $\TheSet$. 
\end{proof}

Algorithm~\ref{alg:setchain-vanilla} guarantees Property~\ref{prop:add-get-local}.

\begin{lem}\label{lem:vanilla-add-get-local}
     Let $\setobject.\add_v(e)$ be an operation invoked on a correct server $v$, and $e$ is valid. Then, eventually all invocations $(T,H,h,P)=\setobject.\get_v()$ satisfy $e \in T$.
\end{lem}

\begin{proof}
The execution of $\setobject.\add_v(e)$ by the correct server $v$ will add $e$ to $\TheSet$ (Line \ref{alg:setchain-vanilla-setaddition}) if it is not already present.
Since elements are never removed from $\TheSet$, element $e$ will eventually appear in $\TheSet$ returned in all future $\setobject.\get_v()$ invocations which return $(\TheSet, \history, \epoch, \proofs)$.
\end{proof}

Algorithm~\ref{alg:setchain-vanilla} guarantees Property~\ref{prop:get-global}.

\begin{lem}\label{lem:vanilla-get-global}
Let $v$ and $w$ be two correct servers, let $e$ be a valid element, and let $(T,H,h,P)=\setobject.\get_v()$. If $e \in T$, then eventually all invocations $(T',H',h',P')=\setobject.\get_w()$ satisfy that $e \in T'$.
\end{lem}

\begin{proof}
%
In Algorithm~\ref{alg:setchain-vanilla}, an element is added by server $v$ to $\TheSet$ in Lines \ref{alg:setchain-vanilla-setaddition} and \ref{alg:setchain-vanilla-othersetaddition} only. 

First, let us consider the case where  $e$ is added by server $v$ to $\TheSet$ in Line~\ref{alg:setchain-vanilla-setaddition} when a client invoked $\setobject.\add_v(e)$. Then, $v$ invokes $\ledgerobject.\Append_v(e)$ in Line~\ref{alg:setchain-vanilla-append}.
From Property~\ref{prop:ledger-add-eventual-notify}, a block $B$ containing $e$ is eventually notified to all correct \setchain servers, including $w$. 

On the other hand, if $e$ is added by server $v$ to $\TheSet$ at Line~\ref{alg:setchain-vanilla-othersetaddition} for the first time, then a block $B$ containing $e$ was notified to $v$ by the ledger $\ledgerobject$. From Property \ref{prop:ledger-consistent-notification}, this block must have been notified to all correct \setchain servers, including $w$.

In both cases, when processing block $B$, server $w$ adds the valid element $e$ to $\TheSet$ (Line~\ref{alg:setchain-vanilla-othersetaddition}) if it is not already present, and is never removed from $\TheSet$. Hence, all future invocations of $\setobject.\get_w()$ return tuples $(T', H', h', P')$ such that $e \in T'$.
\end{proof}

Algorithm~\ref{alg:setchain-vanilla} guarantees Property~\ref{prop:eventual-get}.

\begin{lem}\label{lem:vanilla-eventual-get}
Let $v$ be a correct server, let $e$ be a valid element and let $(T,H,h,P)=\setobject.\get_v()$. If $e \in T$, then eventually all invocations $(T',H',h',P')=\setobject.\get_v()$ satisfy that $e \in H'$.
\end{lem}

\begin{proof}
%
The element $e$ is added by server $v$ to $\TheSet$ in either Line~\ref{alg:setchain-vanilla-setaddition} or \ref{alg:setchain-vanilla-othersetaddition}. 

First, let us consider the case where $e$ is added by $v$ to $\TheSet$ in Line~\ref{alg:setchain-vanilla-setaddition} when a client invoked $\setobject.\add_v(e)$. Then, $v$ invokes $\ledgerobject.\Append_v(e)$ in Line~\ref{alg:setchain-vanilla-append}.
From Property~\ref{prop:ledger-add-eventual-notify}, eventually a block $B$ containing $e$ is notified to $v$. 

Otherwise, if $e$ is added by server $v$ to $\TheSet$ in Line~\ref{alg:setchain-vanilla-othersetaddition} for the first time, then it was because a block $B$ containing $e$ was notified to $v$ by ledger $\ledgerobject$.
    
Then, in both cases, $v$ adds the batch of valid elements $G$ from the block $B$ to $\history$ at Line~\ref{alg:setchain-vanilla-epochaddition}, with $e \in G$. After that, all future invocations of $\setobject.\get_w()$ return tuples $(T', H', h', P')$ that satisfy $e \in H'$.
\end{proof}

Algorithm~\ref{alg:setchain-vanilla} guarantees Property \ref{prop:unique-epoch}.

\begin{lem}\label{lem:vanilla-unique-epoch}
Let $v$ be a correct server, $(T,H,h,P)=\setobject.\get_v()$, and let $i,i' \in \{1,\ldots,h\}$ with $i \neq i'$. Then, $H[i] \cap H[i'] = \emptyset$.
\end{lem}

\begin{proof}
%
By way of contradiction, let us assume that for some valid element $e$ it holds that $e \in \history[i]$ and $e \in \history[i']$. Without loss of generality, let us assume that $1 \leq i < i' \leq \epoch$.

Observe that the elements of an epoch are inserted into $\history$ in Line~\ref{alg:setchain-vanilla-epochaddition}, and they are added in
increasing epoch order (variable $\epoch$ is increased every time an epoch is created (line~\ref{alg:setchain-vanilla-epochinc})and never decreased).
Then,  when $\history[i']$ is mapped to a set $G$ with $e \in G$ in Line~\ref{alg:setchain-vanilla-epochaddition}, it already holds that 
$e \in \history[i]$. But, from Line~\ref{alg:setchain-vanilla-othervalidelement}, $G$ does not contain elements already in $\history$.
Hence we have a contradiction.
\end{proof}



Algorithm~\ref{alg:setchain-vanilla} guarantees Property~\ref{prop:consistent-gets}.

\begin{lem}\label{lem:vanilla-consistent-gets}
Let $v,w$ be correct servers, let $(T,H,h,P)=\setobject.\get_v()$ and $(T',H',h',P')=\setobject.\get_w()$, and let $i \in \{1,\ldots,\min(h,h')\}$. Then $H[i]=H'[i]$.
\end{lem}

\begin{proof}
%


The proof proceeds by induction in the epoch number n.

\begin{itemize}
\item Base case, $n = 0.$ $ H[0] = \emptyset = H'[0].$

\item  Inductive step: we will show that $H[n] = H'[n]$ assuming that $\forall i < n, H[i] = H'[i].$
\end{itemize}

First, we show that $H[n] \subseteq H'[n]$.

Let $e \in H[n]$. 
Then $e$ that was added by $v$ to $H[n]$ at Line \ref{alg:setchain-vanilla-epochaddition} when a ledger block $B$ containing $e$ was processed by $v$.
It is easy to see that $B$ is the $n$-th block received by $v$. 
Therefore, from Property~\ref{prop:ledger-consistent-notification}, $w$ also receives $B$ as the $n$-th block.  
The fact that $v$ added element $e$ to $H[n]$ implies that (1) $e$ is a valid element and (2) $e \notin H[0..n-1].$
By inductive hypothesis and (2), we conclude that $e$ is not in $H'[0..n-1]$. 
Since $e$ is a valid element in the $n$-th block received by $w$ that is not in $H[0..n-1]$, $e$ is added to $H'[n]$.
The proof that $H'[n] \subseteq H[n]$ is analogous.

\end{proof}

Algorithm~\ref{alg:setchain-vanilla} guarantees Property~\ref{prop:add-before-get}.

\begin{lem}\label{lem:vanilla-add-before-get}
Let $v$ be a correct server, $e$ be a valid element, $(T,H,h,P)$ $=\setobject.\get_v()$, and
$e \in T$. Then there was an operation $\setobject.\add_w(e)$ invoked in the past in some server $w$.
\end{lem}

\begin{proof}
%
Elements are added to $\TheSet$ in Lines~\ref{alg:setchain-vanilla-setaddition} and \ref{alg:setchain-vanilla-othersetaddition}. 

First, let us consider the case where $e$ was added by server $v$ to $\TheSet$ in Line~\ref{alg:setchain-vanilla-setaddition}. Then, it was added when operation $\setobject.\add_v(e)$ was being processed.

On the other hand, if $e$ was added by server $v$ to $\TheSet$ in Line~\ref{alg:setchain-vanilla-othersetaddition}, then a ledger block $B$ containing $e$ was notified to $v$ by the ledger $\ledgerobject$. From Property~\ref{prop:notification-implies-append}, some server $w$ invoked $\ledgerobject.\Append_w(e)$.
Recall that, as mentioned in Section~\ref{subsec:systemmodel}, we assume that a server cannot create a valid element by itself and that clients and servers do not collude. So, a server $w$ cannot append a valid element $e$ with $\ledgerobject.\Append_w(e)$ without a client invocation $\setobject.\add_w(e)$.
\end{proof}

Algorithm~\ref{alg:setchain-vanilla} guarantees Property~\ref{prop:valid-epoch}.

\begin{lem}\label{lem:vanilla-valid-epoch}
Let $v$ be a correct server, $(T,H,h,P)=\setobject.\get_v()$, and $i \in \{1, \ldots, h\}$. Then eventually all invocations $(T',H',h',P')=\setobject.\get_v()$ satisfy that $P'$ contains at least $f+1$ epoch proofs of $H[i]$.
\end{lem}

\begin{proof}
%
Let $E=\history[i]$ be an epoch with epoch number $i$, whose valid and new elements as set $G$ were added to $\history$ by $v$ at Line~\ref{alg:setchain-vanilla-epochaddition}. Once this happens, the hash of $G$ is computed and signed by $v$ as $p_v$. Then, an epoch proof consisting of $\langle i,p_v,v \rangle$ is appended to the ledger by $v$ in Line~\ref{alg:setchain-vanilla:epochproof}.

By Property~\ref{prop:ledger-consistent-notification}, all correct \setchain servers will eventually receive enough ledger blocks to reach the creation of epoch $i$.
From lemma~\ref{lem:vanilla-consistent-gets}, all correct \setchain servers agree on the content of the $i$-th epoch.
Therefore, every correct \setchain server $w$ maps $G$ to $\history_w[i]$, generates an epoch proof $\langle i,p_w,w \rangle$ for the i-th epoch, and appends it to the ledger.

Then, from Property~\ref{prop:ledger-add-eventual-notify}, ledger blocks containing the epoch proofs of $\history[i]$ will be notified to all correct \setchain servers, including server $v$. Server $v$ will add these epoch proofs to $\proofs$ in Line~\ref{alg:setchain-vanilla:validepochproof}. Since we assume a system with $n$ servers, where at most $f<n/2$ are not correct, at least $f+1$ epoch proofs for the i-th will be appended to the ledger. Hence, eventually, in all invocations $(T', H', h', P')=\setobject.\get_v()$, it will hold that $P'$ will contain at least $f+1$ epoch proofs of $\history[i]$.
\end{proof}

The combination of all the previous lemmas shows that Algorithm~\ref{alg:setchain-vanilla} implements a \setchain with epoch proofs.

\subsection{Correctness of Algorithm~\ref{alg:setchain-compress}}

This section follows the same pattern as the previous section, proving Properties \ref{prop:consistent-sets} to \ref{prop:valid-epoch} to conclude that Algorithm~\ref{alg:setchain-compress} implements a \setchain with epoch proofs.

Algorithm \ref{alg:setchain-compress} guarantees Property \ref{prop:consistent-sets}.

\begin{lem} \label{lem:compress-consistent-sets}
     Let $(T,H,h,P)=\setobject.\get_v()$ be an invocation to a correct server $v$. Then, $\forall i \in \{1,\ldots,h\}, H[i] \subseteq T$.
\end{lem}

\begin{proof}
We prove that in a correct server $v$ it holds that $\forall i \in [1,\epoch], \history[i] \subseteq \TheSet$ at all times.

Initially, we have that $\epoch=0$, since Line \ref{alg:setchain-compress-epochinc} was never executed. Hence, Line \ref{alg:setchain-compress-epochaddition} was not executed either, $\history$ is empty, and the claim trivially holds initially. 

Let us assume now that $\epoch>0$ and consider some $i \in [1,epoch]$. Since $\epoch$ is only modified in Line \ref{alg:setchain-compress-epochinc} by increments of 1, then, when it became equal to $i$, all the elements $G$ that were later added to $\history$ in epoch $i$ (in Line \ref{alg:setchain-compress-epochaddition}) were guaranteed to be in $\TheSet$ before (in Line \ref{alg:setchain-compress-othersetaddition}), and are never removed from $\TheSet$. 
\end{proof}

Algorithm \ref{alg:setchain-compress} guarantees Property \ref{prop:add-get-local}.

\begin{lem}\label{lem:compress-add-get-local}
   Let $\setobject.\add_v(e)$ be an operation invoked on a correct server $v$, and $e$ is valid. Then, eventually all invocations $(T,H,h,P)=\setobject.\get_v()$ satisfy $e \in T$.
\end{lem}

\begin{proof}
   The execution of $\setobject.\add_v(e)$ by the correct server $v$ will add $e$ to $\TheSet$ (Line \ref{alg:setchain-compress-setaddition}) if not already present, which will eventually be returned in all future $\setobject.\get_v()$ invocations which return $(\TheSet, \history, epoch, \proofs)$.
\end{proof}

Algorithm \ref{alg:setchain-compress} guarantees Property \ref{prop:get-global}.

\begin{lem}\label{lem:compress-get-global}
     Let $v$ and $w$ be two correct servers, let $e$ be a valid element, and let $(T,H,h,P)=\setobject.\get_v()$. If $e \in T$, then eventually all invocations $(T',H',h',P')=\setobject.\get_w()$ satisfy that $e \in T'$.
\end{lem}

\begin{proof}
    
    In Algorithm~\ref{alg:setchain-compress}, an element is added by server $v$ to $\TheSet$ in Lines \ref{alg:setchain-compress-setaddition} and \ref{alg:setchain-compress-othersetaddition} only. 
    
    First, let us consider the case where $e$ is added by server $v$ to $\TheSet$ in Line \ref{alg:setchain-compress-setaddition} when a client invoked $\setobject.\add_v(e)$. Then, $v$ adds $e$ to $batch$ by invoking $\setobject.add\_to\_batch_v(e)$ (Line~\ref{alg:setchain-compress-batchaddition}). 
    %
    %
    Eventually, after $e$ is added to the $batch$, a notification $\mathtt{isReady}(\batch)$ is signaled.
    Therefore, eventually $v$ compresses $batch$ with element $e$ and appends the compressed batch $cb$ to the ledger $\ledgerobject$ (Line~\ref{alg:setchain-compress-append}). Note that, elements in $batch$ are removed only after the compressed version of $batch$ is added to the ledger $\ledgerobject$ (line~\ref{alg:setchain-compress-emptybatch}).
    %
    %
    From Property~\ref{prop:ledger-add-eventual-notify}, eventually a block $B$ containing the compressed batch $cb$ is notified to all correct \setchain servers, including $w$.
    
    On the other hand, if $e$ is added by server $v$ to $\TheSet$ at Line~\ref{alg:setchain-compress-othersetaddition} for the first time, then a block $B$ containing a compressed batch $cb$ with $e \in cb$ was notified to $v$ by the ledger $\ledgerobject$. From Property \ref{prop:ledger-consistent-notification}, this block must have been notified to all correct \setchain servers, including $w$.
    
    In both cases, when processing $cb$ (with $e \in cb$), server $w$ adds the valid element $e$ to $\TheSet$ (Line~\ref{alg:setchain-compress-othersetaddition}) if it is not already present and is never removed from $\TheSet$. Hence, all future invocations of $\setobject.\get_w()$ return tuples $(T', H', h', P')$ with $e \in T$.
\end{proof}

Algorithm \ref{alg:setchain-compress} guarantees Property~\ref{prop:eventual-get}.

\begin{lem}\label{lem:compress-eventual-get}
    Let $v$ be a correct server, let $e$ be a valid element and let $(T,H,h,P)=\setobject.\get_v()$. If $e \in T$, then eventually all invocations $(T',H',h',P')=\setobject.\get_v()$ satisfy that $e \in H'$.
\end{lem}

\begin{proof}
    
    The element $e$ was added to $\TheSet$ of server $v$ in Algorithm \ref{alg:setchain-compress} in either Line \ref{alg:setchain-compress-setaddition} or Line \ref{alg:setchain-compress-othersetaddition}. 
    
    First, let us consider the case where $e$ is added by server $v$ to $\TheSet$ in Line \ref{alg:setchain-compress-setaddition} when a client invoked $\setobject.\add_v(e)$. Then, $v$ adds $e$ to $batch$ by invoking $\setobject.add\_to\_batch_v(e)$ (Line~\ref{alg:setchain-compress-batchaddition}). 
    %
    %
    Eventually, after $e$ is added to the $batch$, a notification $\mathtt{isReady}(\batch)$ is signaled.
    Therefore, eventually $v$ compresses $batch$ with element $e$ and appends the compressed batch $cb$ to the ledger $\ledgerobject$ (Line~\ref{alg:setchain-compress-append}). Note that, elements in $batch$ are removed only after the compressed version of $batch$ is added to the ledger $\ledgerobject$ (line~\ref{alg:setchain-compress-emptybatch}). From Property~\ref{prop:ledger-add-eventual-notify}, eventually a block $B$ containing the compressed batch $cb$ is notified to $v$.

     On the other hand, if $e$ is added to $\TheSet$ at Line \ref{alg:setchain-compress-othersetaddition} for the first time, it was because a block $B$ containing the compressed batch $cb$ was notified to $v$ by the ledger $\ledgerobject$ and $e \in cb$.

     Then, in either case, after receiving $cb$ with $e \in cb$ in block $B$, $v$ decompresses $cb$ and adds the set $G$ of its valid elements to $\history$ (if not there already; Line \ref{alg:setchain-compress-epochaddition}). Since $e$ is valid, after this it holds that $e \in \history$. Then, eventually all invocations $(T', H', h', P')=\setobject.\get_v()$ satisfy that $e \in H'$.
\end{proof}

Algorithm \ref{alg:setchain-compress} guarantees Property~\ref{prop:unique-epoch}.

\begin{lem}\label{lem:compress-unique-epoch}
    Let $v$ be a correct server, $(T,H,h,P)=\setobject.\get_v()$, and let $i,i' \in \{1,\ldots,h\}$ with $i \neq i'$. Then, $H[i] \cap H[i'] = \emptyset$.
\end{lem}

\begin{proof}
    
    By way of contradiction, let us assume that for some valid element $e$ it holds that $e \in \history[i]$ and $e \in \history[i']$. Without loss of generality, let us assume that $1 \leq i < i' \leq \epoch$. 
    
    Observe that the elements of an epoch are inserted into $history$ in Line~\ref{alg:setchain-compress-epochaddition} and they do it in increasing epoch order. Then, when a set $G$ with $e \in G$ is mapped to $\history[i']$ in Line~\ref{alg:setchain-compress-epochaddition}, it already holds that  $e \in \history[i]$. But, from Line~\ref{alg:setchain-compress-othervalidelement}, $G$ cannot contain elements already in $\history$. Hence we have a contradiction.
\end{proof}

Algorithm \ref{alg:setchain-compress} guarantees Property~\ref{prop:consistent-gets}.

\begin{lem}\label{lem:compress-consistent-gets}
    Let $v,w$ be correct servers, let $(T,H,h,P)=\setobject.\get_v()$ and $(T',H',h',P')=\setobject.\get_w()$, and let $i \in \{1,\ldots,\min(h,h')\}$. Then $H[i]=H'[i]$.
\end{lem}

\begin{proof}


    The proof proceeds by induction in the epoch number n.
    
    \begin{itemize}
    \item Base case, $n = 0.$ $ H[0] = \emptyset = H'[0].$
    
    \item  Inductive step: we will show that $H[n] = H'[n]$ assuming that $\forall i < n, H[i] = H'[i].$
    \end{itemize}

    First, we show that $H[n] \subseteq H'[n]$.
    
    Let $e \in H[n]$. 
    Then $e$ that was added by $v$ to $H[n]$ at Line \ref{alg:setchain-compress-epochaddition} when a ledger block $B$ containing the compressed batch $cb$ was processed by $v$ and $e \in cb$.
    %
    %
    From Property~\ref{prop:ledger-consistent-notification}, we know that $w$ receives the same set of blocks in the same order as $v$. So, it receives $B$ and processes $cb$ in the same order as $v$ and $e \in cb$  
    The fact that $v$ added element $e$ to $H[n]$ implies that (1) $e$ is a valid element and (2) $e \notin H[0..n-1].$
    By inductive hypothesis and (2), we conclude that $e$ is not in $H'[0..n-1]$. 
    Since $e$ is a valid element that is not in $H[0..n-1]$, $e$ is added to $H'[n]$.
    The proof that $H'[n] \subseteq H[n]$ is analogous.
    
\end{proof}

Algorithm \ref{alg:setchain-compress} guarantees Property~\ref{prop:add-before-get}.

\begin{lem}\label{lem:compress-add-before-get}
    Let $v$ be a correct server, $e$ be a valid element, $(T,H,h,P)$ $=\setobject.\get_v()$, and
    $e \in T$. Then there was an operation $\setobject.\add_w(e)$ invoked in the past in some server $w$.
\end{lem}

\begin{proof}

    Elements are added to $\TheSet$ in the Lines~\ref{alg:setchain-compress-setaddition} and \ref{alg:setchain-compress-othersetaddition}. 
    
    First, let us consider the case where $e$ was added by server $v$ to $\TheSet$ in Line~\ref{alg:setchain-compress-setaddition}. Then, it was added when operation $\setobject.\add_v(e)$ was being processed.
    
    On the other hand, if $e$ was added by server $v$ to $\TheSet$ in Line~\ref{alg:setchain-compress-othersetaddition}, then a ledger block $B$ containing a compressed batch $cb$ was notified to $v$ by the ledger $\ledgerobject$, and $e \in cb$. From Property~\ref{prop:notification-implies-append}, some server $w$ invoked $\ledgerobject.\Append_w(cb)$.
    
    Recall that, as mentioned in Section~\ref{subsec:systemmodel}, we assume that a server cannot create a valid element by itself, and clients and servers do not collude. So, a server $w$ cannot append a valid element $e$ as a part of the compressed batch $cb$ with $\ledgerobject.\Append_w(cb)$ without a client invocation $\setobject.\add_w(e)$.
\end{proof}

Algorithm \ref{alg:setchain-compress} guarantees Property~\ref{prop:valid-epoch}.

\begin{lem}\label{lem:compress-valid-epoch}
    Let $v$ be a correct server, $(T,H,h,P)=\setobject.\get_v()$, and $i \in \{1, \ldots, h\}$. Then eventually all invocations $(T',H',h',P')=\setobject.\get_v()$ satisfy that $P'$ contains at least $f+1$ epoch proofs of $H[i]$.
\end{lem}

\begin{proof}

    Let $\history[i]$ be an epoch with epoch number $i$ whose valid and new elements as set $G$ were added to $\history$ by $v$ at Line~\ref{alg:setchain-compress-epochaddition}. Once this happens, the hash of $G$ is computed and signed by $v$ as $p_v$. Then, an epoch proof consisting of $\langle i,p_v,v \rangle$ is added to the $batch$ by $v$ in Line~\ref{alg:setchain-compress:epochproof}. Once the $batch$ is ready, it will be compressed and the compressed batch $cb'$ will be appended to the ledger $\ledgerobject$ by $v$ (line~\ref{alg:setchain-compress-append}).
    
    From Property~\ref{prop:ledger-consistent-notification}, ledger $\ledgerobject$ notifies the same set of blocks in the same order to all servers. Eventually, the ledger block $B$ with the compressed batch $cb$ was notified to each correct server $w$ as it happened to $v$. Then, server $w$ maps valid elements of $cb$ as $G$ to $\history[i]$, generates an epoch proof $\langle i,p_w,w \rangle$ for the i-th epoch, and appends it to its $batch$. Eventually, this batch will be ready and compressed, and the compressed batch $cb''$ will be appended to the ledger $\ledgerobject$. 
    
    Then, from Property~\ref{prop:ledger-add-eventual-notify}, ledger blocks containing the compressed batches with the epoch proofs of the i-th epoch will be notified to all correct \setchain servers, including server $v$. Server $v$ will add these epoch proofs to $\proofs$ in Line~\ref{alg:setchain-compress:validepochproof}. Since we assume a system with $n$ servers, where at most $f<n/2$ are not correct, at least $f+1$ epoch proofs for $\history[i]$ will be appended by correct servers to the ledger as a part of compressed batches. Hence, eventually, in all invocations $(T', H', h', P')=\setobject.\get_v()$, it will hold that $P'$ will contain at least $f+1$ epoch proofs of $\history[i]$.
\end{proof}

\subsection{Correctness of Algorithm~\ref{alg:setchain-hash}}

Before writing the proofs for the properties of \setchain, we need some lemmas for \emph{Epoch Consolidation} in Hashchain.
Lemma \ref{lem:hash-batch-original-retrieval} states that when a correct server generates a hash batch and appends it to the ledger, that server is available to share the original batch to other servers when requested.

\begin{lem}\label{lem:hash-batch-original-retrieval}
    Let $v$ be a correct server and let $hb$ be a valid hash batch appended by $v$ to the ledger L through $\texttt{L.append}(hb)$, where $hb = \langle hs,sg,v \rangle$. Then $v$ is available to share the original batch $batch\_original$ corresponding to the hash $hs$ when requested by other servers $w$.
\end{lem}
\begin{proof}

    Two places where a hash batch is appended to the Ledger $\ledgerobject$ in Algorithm \ref{alg:setchain-hash} are at Lines \ref{alg:setchain-hash-append} and  \ref{alg:setchain-hash-otherappend}.

    When a correct server $v$ appends a hash batch $hb$ to the ledger $\ledgerobject$ at Line \ref{alg:setchain-hash-append}, it stores the hash $hs$ and the corresponding $batch\_original$ in local, by calling the $\mathtt{Register\_batch}(h,\batch)$ function (Line \ref{alg:register-hash-hashtobatch}).
    
    When a correct server $v$ appends a hash batch $hb$ to the ledger $\ledgerobject$ at Line \ref{alg:setchain-hash-otherappend}, it has already either retrieved the $batch\_original$ from its local storage (Line \ref{alg:setchain-hash-retrieveoriginalbatchlocal}) or has requested the original batch using the hash $hs$ found in $hb$ (Line \ref{alg:setchain-hash-requestoriginalbatch}), if the original batch is not found locally. $v$ proceeds to append the hash batch $hb$ to the Ledger $\ledgerobject$, only after verifying (Line \ref{alg:setchain-hash-batchverify}) that the $batch\_original$ is not empty and the hash of $batch\_original$ matches the hash $hs$ found in $hb$. Also, before appending, it stores the hash $hs$ and the corresponding $batch\_original$ in local, by calling the $\mathtt{Register\_batch}(h,\batch)$ function (Line \ref{alg:other-register-hash-hashtobatch}).

    So, in both cases, $v$ will be available to share the original batch $batch\_original$ corresponding to the hash $hs$ when requested by other servers $w$.
\end{proof}

\begin{lem}\label{lem:hash-epoch-consolidation-liveness}
    Let $v$ be a correct server and let $hb$ be a valid hash batch appended by $v$ to the ledger L through $\texttt{L.append}(hb)$, where $hb = \langle hs,sg,v \rangle$. Then eventually $|hash\_to\_signers[hs]| \geq f+1$ in each correct server $w$. 
\end{lem}

\begin{proof}
    When hash batch $hb$ is appended to the ledger $\ledgerobject$ by the correct server $v$, by Property \ref{prop:ledger-add-eventual-notify} a block $B$ containing $hb$ will eventually be notified to all correct \setchain~servers.

    Each correct \setchain~server $w$, when processing $hb$ as a part of the received block $B$, requests the original batch of elements from $v$ after validating the signature $sg$. After receiving the original batch $batch\_original$, $w$ verifies the hash $hs$ as $\texttt{Hash}(batch\_original) = hs$. Since the hash is valid ($v$ is correct), then $w$ generates a signature $sg'$ of the hash $hs$. Then, it generates a new hash batch $hb' = \langle hs,sg',w \rangle$ and appends $hb'$ to the ledger. Server $w$ also adds the identity $v$ to $hash\_to\_signers[hs]$.

    Again, by Property \ref{prop:ledger-add-eventual-notify}, all the correct \setchain~servers will eventually be notified of blocks from the ledger with hash batches for $hs$ from all correct servers. When the hash batch of a correct server is processed, the server is added to $hash\_to\_signers[hs]$.
    
    Since we assume a system with at least $f+1$ correct servers, eventually $|hash\_to\_signers[hs]| \geq f+1$.
\end{proof}

Algorithm \ref{alg:setchain-hash} guarantees Property \ref{prop:consistent-sets}. 

\begin{lem} \label{lem:hash-consistent-sets}
     Let $(T,H,h,P)=\setobject.\get_v()$ be an invocation to a correct server $v$. Then, $\forall i \in \{1,\ldots,h\}, H[i] \subseteq T$.
\end{lem}

\begin{proof}
We prove that in a correct server $v$ it holds that $\forall i \in [1,epoch], \history[i] \subseteq \TheSet$ at all times.

Initially, we have that $\epoch=0$, since Line \ref{alg:setchain-hash-epochinc} was never executed. Hence, Line \ref{alg:setchain-hash-epochaddition} was not executed either, $\history$ is empty, and the claim trivially holds. 

Let us assume now that $\epoch>0$ and consider some $i \in [1,epoch]$. Since $\epoch$ is only modified in Line \ref{alg:setchain-hash-epochinc} by increments of 1, when it became equal to $i$, all the elements $G$ that were later (in Line \ref{alg:setchain-hash-epochaddition}) added to $\history$ in epoch $i$ were guaranteed to be in $\TheSet$ before (in Line \ref{alg:setchain-hash-othersetaddition}). 
\end{proof}

Algorithm \ref{alg:setchain-hash} guarantees Property \ref{prop:add-get-local}.

\begin{lem}\label{lem:hash-add-get-local}
   Let $\setobject.\add_v(e)$ be an operation invoked on a correct server $v$, and $e$ is valid. Then, eventually all invocations $(T,H,h,P)=\setobject.\get_v()$ satisfy $e \in T$.
\end{lem}

\begin{proof}
   The execution of $\setobject.\add_v(e)$ by the correct server $v$ will add $e$ to $\TheSet$ (Line \ref{alg:setchain-hash-setaddition}) if not already present. Since the elements are never removed from $the\_set$, $e$ will eventually be returned in all future $\setobject.\get_v()$ invocations which return $(\TheSet, \history, epoch, \proofs)$.
\end{proof}

Algorithm \ref{alg:setchain-hash} guarantees Property \ref{prop:get-global}.

\begin{lem}\label{lem:hash-get-global}
     Let $v$ and $w$ be two correct servers, let $e$ be a valid element, and let $(T,H,h,P)=\setobject.\get_v()$. If $e \in T$, then eventually all invocations $(T',H',h',P')=\setobject.\get_w()$ satisfy that $e \in T'$.
\end{lem}

\begin{proof}
    
    Two places where an element is added to $\TheSet$ in Algorithm \ref{alg:setchain-hash} are in Lines \ref{alg:setchain-hash-setaddition} and \ref{alg:setchain-hash-othersetaddition}. 
    
    Let us consider the case where in server $v$, $e$ is added to $\TheSet$ in Line \ref{alg:setchain-hash-setaddition} when a client invoked $\setobject.\add_v(e)$. Then $v$ adds $e$ to $batch$ by invoking $\setobject.add\_to\_batch_v(e)$ (Line~\ref{alg:setchain-hash-batchaddition}). Once a batch is ready, $v$ invokes $\ledgerobject.\Append(hb)$, which appends a hash batch $hb$ to the ledger $\ledgerobject$, where $hb = \langle hs,sg,v \rangle$, $hs = \Hash(\batch)$, and $e \in batch$.

    On the other hand, if $e$ is added to $\TheSet$ at Line \ref{alg:setchain-hash-othersetaddition} for the first time in server $v$, then a block $B$ containing $hb$ was notified to $v$ by the ledger, and the original batch was recovered and is valid. 
    Since the hash is valid, $v$ signs it and generates a new hash batch as $hb = \langle hs,sg,v \rangle$ and appends $hb$ to the Ledger and $e \in batch\_original$.
    
    In either case, according to Property \ref{prop:ledger-add-eventual-notify}, eventually a block $B$ containing $hb$ is notified to all correct \setchain~servers, including $w$.  When $w$ encounters the hash batch with the hash $hs$ as a part of block $B$, eventually $w$ could retrieve the original batch from $v$ as proved in Lemma \ref{lem:hash-batch-original-retrieval}, if it doesn't have it already. After verifying the original batch corresponding to $hb$, $w$ would add the valid element $e$ to $\TheSet$ if it is not already present (Line \ref{alg:setchain-hash-othersetaddition}). So, eventually all future invocations of $\setobject.\get_w()$ returns $(\TheSet, \history, epoch, \proofs)$ and $e \in \TheSet$.

\end{proof}

Algorithm \ref{alg:setchain-hash} guarantees Property \ref{prop:eventual-get}.

\begin{lem}\label{lem:hash-eventual-get}
    Let $v$ be a correct server, let $e$ be a valid element and let $(T,H,h,P)=\setobject.\get_v()$. If $e \in T$, then eventually all invocations $(T',H',h',P')=\setobject.\get_v()$ satisfy that $e \in H'$.
\end{lem}

\begin{proof}
    
    Two places where an element is added to $\TheSet$ in Algorithm \ref{alg:setchain-hash} are in Lines \ref{alg:setchain-hash-setaddition} and \ref{alg:setchain-hash-othersetaddition}. 
    
    Let us consider the case where in server $v$, $e$ is added to $\TheSet$ in Line \ref{alg:setchain-hash-setaddition} when a client invoked $\setobject.\add_v(e)$. Then $v$ adds $e$ to $batch$ by invoking $\setobject.add\_to\_batch_v(e)$ (Line~\ref{alg:setchain-hash-batchaddition}). Once a batch is ready, $v$ invokes $\ledgerobject.\Append(hb)$, which appends a hash batch $hb$ to the Ledger, where $hb = \langle hs,sg,v \rangle$, $hs = \Hash(\batch)$, and $e \in batch$. 

    On the other hand, if $e$ is added to $\TheSet$ at Line \ref{alg:setchain-hash-othersetaddition} for the first time in server $v$, then a block $B$ containing $hb'= \langle hs,sg',w \rangle$ was notified to $v$ by the ledger, and the original batch was recovered and is valid. 
    Since the hash is valid, then $v$ signs the hash $hs$ and generates a new hash batch for the hash $hs$ as $hb = \langle hs,sg,v \rangle$ and appends $hb$ to the Ledger. Also, it adds $e$ to $\TheSet$. 

    According to Lemma \ref{lem:hash-epoch-consolidation-liveness}, eventually  $|\hashToSigners[hs]| \geq f+1$ in $v$. When $|\hashToSigners[hs]| == f+1$, the batch consolidates and it is assigned an epoch number (Line \ref{alg:setchain-hash-epochinc}). Then, the valid elements from the batch, including $e$, are added to $\history$ as an epoch (Line \ref{alg:setchain-hash-epochaddition}), if not already present. Eventually all invocations $(T', H', h', P')=\setobject.\get_v()$ satisfy that $e \in H'$.
     
\end{proof}

Algorithm \ref{alg:setchain-hash} guarantees Property \ref{prop:unique-epoch}.

\begin{lem}\label{lem:hash-unique-epoch}
    Let $v$ be a correct server, $(T,H,h,P)=\setobject.\get_v()$, and let $i,i' \in \{1,\ldots,h\}$ with $i \neq i'$. Then, $H[i] \cap H[i'] = \emptyset$.
\end{lem}

\begin{proof}
    
    By way of contradiction, let us assume that for some valid element $e$ it holds that $e \in \history[i]$ and $e \in \history[i']$. Without loss of generality, let us assume that $1 \leq i < i' \leq \epoch$. 
    
    So, for $e$ to be included in the epoch $i'$ as a part of $G$ (Line \ref{alg:setchain-hash-epochaddition}), 
    $G$ must have been processed through Line \ref{alg:setchain-hash-othervalidelement}, which allows only elements that are not already present in $\history$ to pass through. We know that $i<i'$ and $epoch$ is only increased (Line \ref{alg:setchain-hash-epochinc}) and never decreased. So, if epoch $i$ was processed first and, as a result, $e \in \history[i]$, then the attempt to add $e$ to $\history[i']$ would not clear the conditions in Line \ref{alg:setchain-hash-othervalidelement}. So, this contradicts our assumption that $e \in \history[i']$. Then, we have $\history[i] \cap \history[i'] = \emptyset$.
\end{proof}

Algorithm \ref{alg:setchain-hash} guarantees Property \ref{prop:consistent-gets}.

\begin{lem}\label{lem:hash-consistent-gets}
    Let $v,w$ be correct servers, let $(T,H,h,P)=\setobject.\get_v()$ and $(T',H',h',P')\\=\setobject.\get_w()$, and let $i \in \{1,\ldots,\min(h,h')\}$. Then $H[i]=H'[i]$.
\end{lem}

\begin{proof}

    The proof proceeds by induction in the epoch number n.
    
    \begin{itemize}
    \item Base case, $n = 0.$ $ H[0] = \emptyset = H'[0].$
    
    \item  Inductive step: we will show that $H[n] = H'[n]$ assuming that $\forall i < n, H[i] = H'[i].$
    \end{itemize}
    
    First, we show that $H[n] \subseteq H'[n]$.
    
    Let $e \in H[n]$. 
    Then $e$ that was added by $v$ to $H[n]$ at Line \ref{alg:setchain-hash-epochaddition} when $v$ received the $f+1$th signature for the hash $hs$ as $hb=<hs,sg,z>$, such that  $v.hash\_to\_signers[hs]=batch\_original\_hs$, $e \in batch\_original\_hs$, $\Hash(\batchOriginal) == hs$ and $|\hashToSigners[hs]| == f+1$.
    The fact that $v$ added element $e$ to $H[n]$ implies that (1) $e$ is a valid element and (2) $e \notin H[0..n-1].$
    Then, by Property~\ref{prop:ledger-add-eventual-notify}, $w$ will eventually receive a ledger block containing $hb=<hs,sg,z>$ and by Properties~\ref{prop:ledger-add-eventual-notify} and \ref{prop:ledger-consistent-notification}, it will be the $f+1$th signature $w$ receives for hash $hs$. 
    Since $w$ has received $f+1$ signatures of $hs$, at least one of the signers is correct, and by Lemma \ref{lem:hash-batch-original-retrieval}, a correct server $z$ was available to share $batch\_original\_hs$ with $w$ when $w$ requested it at Line \ref{alg:setchain-hash-requestoriginalbatch} (if $w$ already doesn't have $batch\_original$) and $e \in batch\_original$.
    %
    %
    By inductive hypothesis and (2), we conclude that $e$ is not in $H'[0..n-1]$. 
    Since $e$ is a valid element that is not in $H'[0..n-1]$, it is added to $H'[n]$ at Line \ref{alg:setchain-hash-epochaddition} by $w$.
    The proof that $H'[n] \subseteq H[n]$ is analogous.   
    
\end{proof}

Algorithm \ref{alg:setchain-hash} guarantees Property \ref{prop:add-before-get}.

\begin{lem}\label{lem:hash-add-before-get}
    Let $v$ be a correct server, $e$ be a valid element, $(T,H,h,P)$ $=\setobject.\get_v()$, and $e \in T$. Then, there was an operation $\setobject.\add_w(e)$ invoked in the past in some server $w$.
\end{lem}

\begin{proof}
    
    Elements are added to $\TheSet$ in the Lines~\ref{alg:setchain-hash-setaddition} and \ref{alg:setchain-hash-othersetaddition}. First, let us consider the case where $e$ was added by server $v$ to $\TheSet$ in Line~\ref{alg:setchain-hash-setaddition}. Then, it was added when operation $\setobject.\add_v(e)$ was being processed.
    
    On the other hand, if $e$ was added by server $v$ to $\TheSet$ in Line~\ref{alg:setchain-hash-othersetaddition}, then a ledger block $B$ containing a hash batch $hb$ was notified to $v$ by the ledger $\ledgerobject$, where $hb = \langle hs,sg,v \rangle$, $hs = \Hash(\batch)$, and $e \in batch$. From Property~\ref{prop:notification-implies-append}, some server $w$ invoked $\ledgerobject.\Append_w(hb)$.
    
    Recall that, as mentioned in Section~\ref{subsec:systemmodel}, we assume that a server cannot create a valid element by itself and do not collude with clients. So, a server $w$ cannot append a valid element $e$ as a part of the hash batch $hb$ with $\ledgerobject.\Append_w(cb)$ without a client invocation $\setobject.\add_w(e)$.
\end{proof}

Algorithm \ref{alg:setchain-hash} guarantees Property \ref{prop:valid-epoch}.

\begin{lem}\label{lem:hash-valid-epoch}
    Let $v$ be a correct server, $(T,H,h,P)=\setobject.\get_v()$, and $i \in \{1, \ldots, h\}$. Then eventually all invocations $(T',H',h',P')=\setobject.\get_v()$ satisfy that $P'$ contains at least $f+1$ epoch proofs of $H[i]$.
\end{lem}

\begin{proof}

    Let $E=\history[i]$ be an epoch with epoch number $i$ added to $\history$ by $v$ at Line~\ref{alg:setchain-hash-epochaddition}. Then $v$ received the $f+1$th signature for the hash $hs$ corresponding to the epoch $i$ as $hb=<hs,sg,z>$, such that  $v.hash\_to\_signers[hs]=batch\_original$, $e \in batch\_original$, $\Hash(\batchOriginal) == hs$ and \\ $|\hashToSigners[hs]| == f+1$. From Properties~\ref{prop:ledger-add-eventual-notify} and \ref{prop:ledger-consistent-notification}, ledger $\ledgerobject$ notifies the same set of blocks in the same order to all servers and the transactions inside a ledger block are ordered as well. So, all correct \setchain servers $w$, eventually receive enough ledger blocks to reach the consolidation of epoch $i$. 

    Once the epoch is consolidated, each correct server $w$ generates the proof for the epoch $E$ as $<i,p_w,w>$ and append it to the $batch$. Eventually, \textbf{upon}$((\mathtt{isReady}(\batch))$ is signalled and $w$ hashes $batch$, signs it and appends it to the ledger as $ohb=<ohs,osg,w>$. By Property \ref{prop:ledger-add-eventual-notify}, eventually $v$ receives a block $B$ containing $ohb$ and by Lemma \ref{lem:hash-batch-original-retrieval}, $v$ can retrieve the $batch$ from $w$, which contains the epoch proof $<i,p_w,w>$. $v$ adds the proof of epoch $E$ to $proofs$ in Line \ref{alg:setchain-hash:validepochproof}. Since we assume a system with $n$ servers, where at most $f<n/2$ are not correct, eventually $v$ will receive and add to $proofs$ at least $f+1$ epoch proofs for epoch $E$, and $proofs$ will be appended to the ledger as a part of the hash batches. Hence, eventually, in all invocations $(T', H', h', P')=\setobject.\get_v()$, it will hold that $P'$ will contain at least $f+1$ epoch proofs of $E$.
\end{proof}
\section{Analytical Performance Study.}
\label{sec:analysis}
%
%
In this section, we present a brief analysis of the stationary throughput that each of these algorithms can achieve.
Let us 
assume in this analysis that all the $n$ system servers are correct. Then, each epoch will have $n$ epoch-proofs, appended to the ledger $\ledgerobject$ in the case of Algorithm~\ref{alg:setchain-vanilla} and sent to the collector to be added to the \batch in the case of \ref{alg:setchain-compress} and~\ref{alg:setchain-hash}. 
Let us assume epoch-proofs have length $l_p$, the elements added by the clients have length $l_e$, and the ledger $\ledgerobject$ has blocks of capacity $C$. With Algorithm~\ref{alg:setchain-vanilla} the valid elements in each ledger block form an epoch. Then, in the steady state, each block will contain $n$ epoch proofs and up to $(C-n \cdot l_p)/l_e$ elements\footnote{For simplicity we will omit floors and ceilings in this analysis.}. If ledger $\ledgerobject$ creates blocks at a rate $R$ (in blocks/second), \ref{alg:setchain-vanilla} can reach a throughput of $T_v=R(C-n \cdot l_p)/l_e$ elements/second.
Let us now consider \ref{alg:setchain-compress} with collector size $c$. In the steady state, for each epoch $n$ epoch proofs are generated. This means that, on average, there are $n$ epoch proofs in each batch. Let us assume the compression algorithm used has a compression ratio $r$. Then, the length in the ledger of an epoch created from a full collector is $\ell=((c-n) \cdot l_e + n \cdot l_p)/{r}$. Then, in the steady state, each ledger block will contain up to $(c-n) \cdot C/ \ell$ valid elements, and \ref{alg:setchain-compress}
can reach a throughput of 
$T_c=\frac{R \cdot (c-n) \cdot C}{\ell}=\frac{R \cdot (c-n) \cdot C}{((c-n) \cdot l_e + n \cdot l_p)/{r}}$ el/s.
Finally, let us consider \ref{alg:setchain-hash} with collector size $c$, and let $l_h$ be the length of a hash-batch. Let us assume that the bottleneck of Algorithm~\ref{alg:setchain-hash} is appending to the ledger $\ledgerobject$. Observe that $n$ hash-batches are appended for each epoch consolidated. Then, in the steady state, \ref{alg:setchain-hash} can reach a throughput of $T_h=R \cdot (c-n) \cdot C/ (n \cdot l_h)$ elements/second.

\subsection{Analysis} 
\label{sec:results-analysis}

We can estimate the achievable throughput for each algorithm using the analysis of Section~\ref{sec:analysis} with the parameters of the evaluation scenario described in Section~\ref{sec:implementation}. These parameters for Algorithm \ref{alg:setchain-vanilla} are, in the average case, $n=10$, $C=500,000$ bytes, $l_e=438$ bytes, $l_p=139$ bytes, and $R=0.8$ block/s. Hence,
$T_v
\approx 955$ 
el/s. 
For Algorithm \ref{alg:setchain-compress} we found that for collector size $c=100$ the compression ratio is roughly $r=2.7$. Hence,  we have
$T_c[c=100] \approx
2,497$ el/s. For collector size $c=500$ the compression ratio is roughly $r=3.5$, and we have 
$T_c[c=500] \approx 3,330$ el/s.
Finally, for Algorithm \ref{alg:setchain-hash}, using that the hash-batch has length $l_h=139$ bytes, for collector size 
$c=100$ we have
$T_h[c=100] \approx
27,157$ el/s, while for collector size $c=500$ we have 
$T_h[c=500] \approx 147,857$ el/s.

Observe that $T_h[c=500]/T_v \approx 155$ and
$T_h[c=500]/T_c[c=500] \approx 44$. Hence, with Hashchain we expect the throughput to increase significantly.
%

\section{Cometbft}\label{sec:cometbft}
CometBFT \cite{tendermint.design} (previously known as Tendermint) is a Byzantine-tolerant state machine replication engine. CometBFT is a blockchain middleware that supports replicating arbitrary applications, written in any programming language.
Two fundamental components of CometBFT are a blockchain \textit{consensus engine} and a
generic \textit{application interface}.
The consensus engine is called \emph{TendermintCore}~\cite{Buchman.2018.Tendermint} and ensures that every validator (server) agrees on the same sequence of blocks with the same set of transactions in the same order.
The application interface called \emph{Application BlockChain Interface (ABCI)}, bridges the consensus engine and the application.
We write the code for the \setchain algorithms in the ABCI section of the ledger.

\subsection{Mapping the Block-based Ledger to CometBFT.}
%
The block-based ledger has two endpoints: an \APPEND function and a \NEWBLOCK notification. Here we define how these endpoints map to CometBFT.

\begin{itemize}
    \item \APPEND: In our algorithms, when a client appends a transaction, it uses the function \BroadcastTxAsync to send the transaction to a CometBFT ledger server.
    This function sends the transaction to the server without waiting for a reply. 
    The server stores the transaction in the mempool and checks if the transaction is valid before sharing it with the other servers using a gossip protocol.

    \item \NEWBLOCK: CometBFT's application interface ABCI has a function \FinalizeBlock. When the CometBFT validators agree upon the proposed block and finalize the order of transactions, the block is sent to all the CometBFT servers for the application layer to process the block. This allows the application to generate additional data (e.g., events, updates to the application's state, etc.) after the block has been finalized. So, the process done in the \setchain algorithms upon \NEWBLOCK notification is done as a part of CometBFT's \FinalizeBlock application logic.
\end{itemize}
\begin{table}[t!]
    \centering
    \caption{Parameters for \setchain evaluation}
    \label{tab:setchain_param}
    \begin{tabular}{|c|c|c|}
        \hline
        \textbf{Name} & \textbf{Description}& \textbf{Values}\\
        \hline 
        \hline
        $sending\_rate$ & Adding rate (el/s)
        & $10000$, $5000$, $1000$, $500$ \\
         \hline
         $collector\_limit$& Collector size (el)& 100, 500\\
         \hline
         $server\_count$& Number of servers & 4, 7, 10\\
         \hline
         $network\_delay$& Delay increase (ms)&0, 30, 100\\
         \hline
    \end{tabular}
\end{table}

\begin{table}[t!]
    \centering
    \caption{Throughput Comparison (upto 50s) for Figure \ref{fig:throughput_comparison}}
    \label{tab:throughput}
    \begin{tabular}{|c|c|c|c|c|}
        \hline
        \textbf{Algorithms} & \textbf{Left} & \textbf{Center} & \textbf{Right}\\
        \hline \hline
         Vanilla&  $171$ el/s & $100$ el/s & $100$ el/s\\
         \hline
         Compresschain&  $996$ el/s & $571$ el/s & $743$ el/s\\
         \hline
         Hashchain&  $4183$ el/s & $2540$ el/s & $7369$ el/s\\
         \hline
    \end{tabular}
\end{table}

\section{Commit Time Comparison} \label{sec:commit_time_comparison}

In this section, we explore how elements get committed over time in the same scenarios explored in
Section~\ref{subsec:efficiencyfactor}. For that, we compute the commit time of the first element, followed by the 10\%, 20\%,
30\%, 40\%, and 50\% of the elements for each algorithm and scenario.
We plot this data in Fig.~\ref{fig:impactofnofserverscommittime} with the $y$ axis truncated for visibility. 
In Fig.~\ref{fig:impactofsendratecommittime} we present results for different sending rates. We observe that Vanilla commits the first element earlier than the other two algorithms. For the low rates ($500$ and $1,000$) elements are committed at a regular pace. However, in the seven combinations of rate and algorithm that showed low efficiency, the pace is not regular.
Fig.~\ref{subfig:impactofnofserverssr10000committime} shows the impact of the number of servers in the commit time. A higher number of servers increases the commit time for Vanilla (although barely observable in the figure) and Compresschain, possibly because it makes consensus harder to reach. Surprisingly, for Hashchain (especially for a collector size of $500$) the pattern seems to be the opposite. We believe that a larger number of servers helps with the reverse hashing process of Hashchain.
Regarding the impact of the network delay, shown in Fig.~\ref{subfig:impactofnetworkdelaysr10000committime}, the conclusion is that it negatively affects the commit times, as expected.

\begin{figure*}[t!]
\begin{subfigure}[b]{0.45\textwidth}
\centering
\includegraphics[trim={1cm 3cm 6cm 0cm},clip,width=\textwidth]{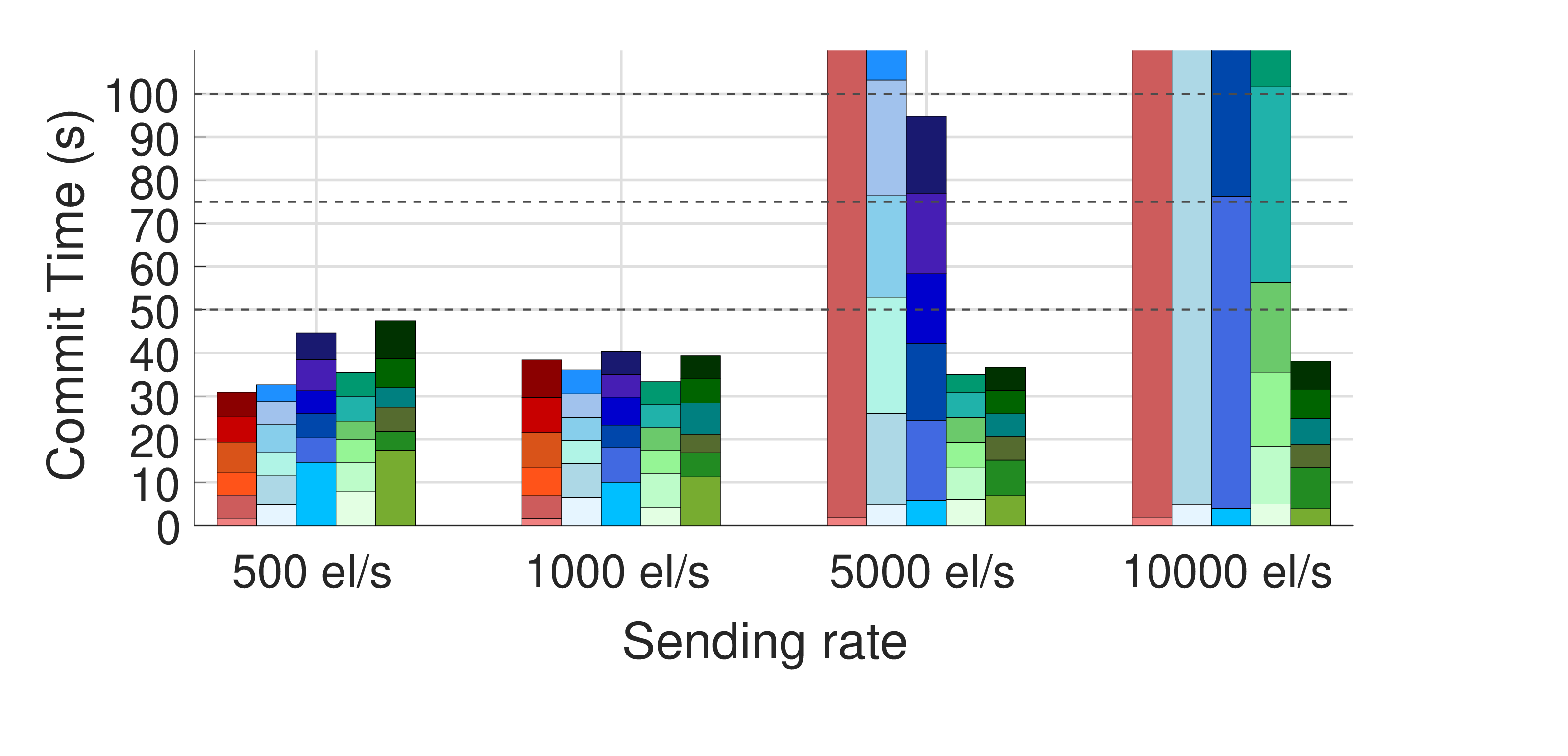}
\caption{Impact of sending rate ($10$ servers, $0$ network delay).}
\label{fig:impactofsendratecommittime}
\end{subfigure}
\hfill
\begin{subfigure}[b]{0.45\textwidth}
\centering
\includegraphics[trim={1cm 3cm 6cm 2cm},clip,width=\textwidth]{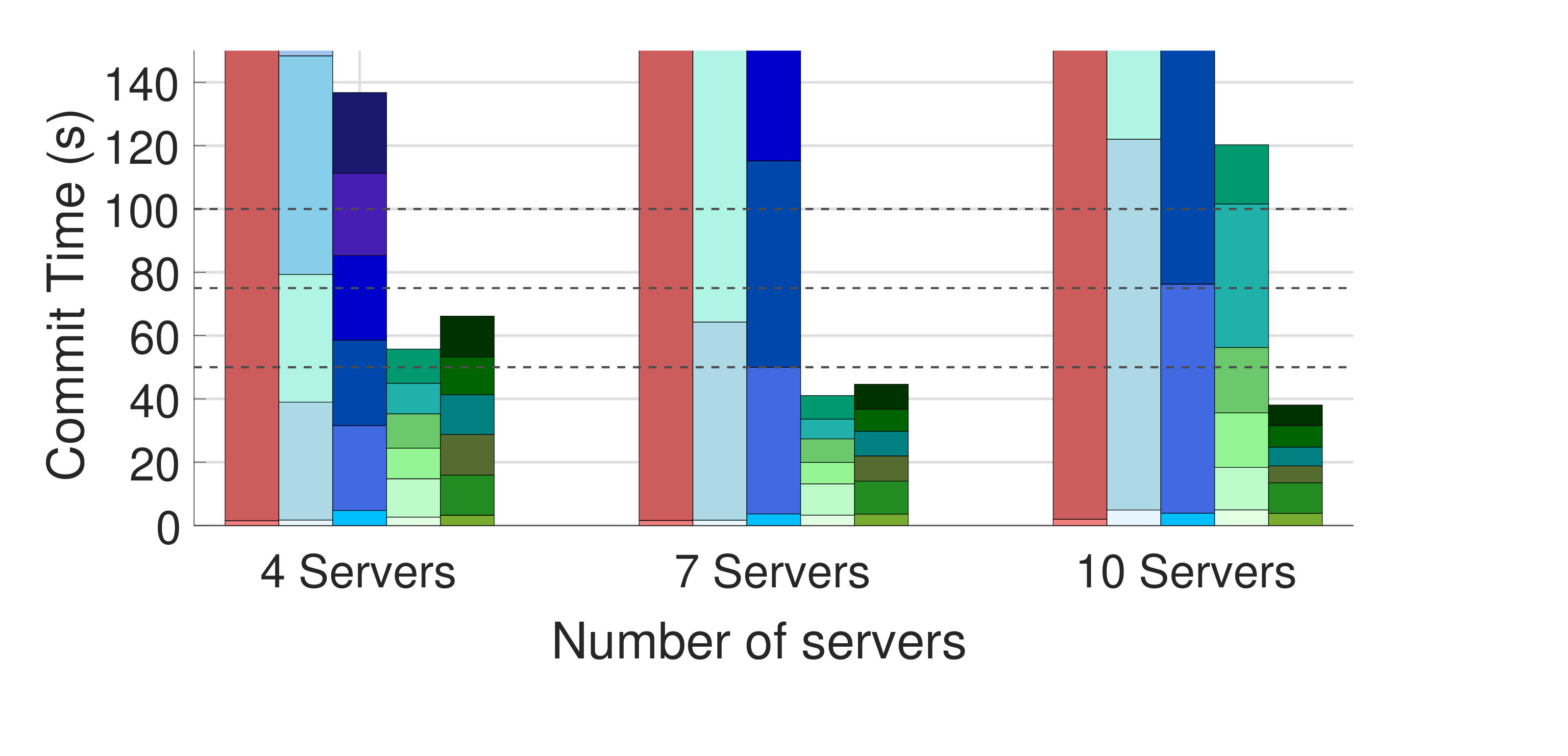}
\caption{Impact of number of servers ($10,000$ el/s, $0$ network delay).}
\label{subfig:impactofnofserverssr10000committime}
\end{subfigure}
\\
\begin{subfigure}[b]{0.45\textwidth}
\centering
\includegraphics[trim={1cm 3cm 6cm 1cm},clip,width=\textwidth]{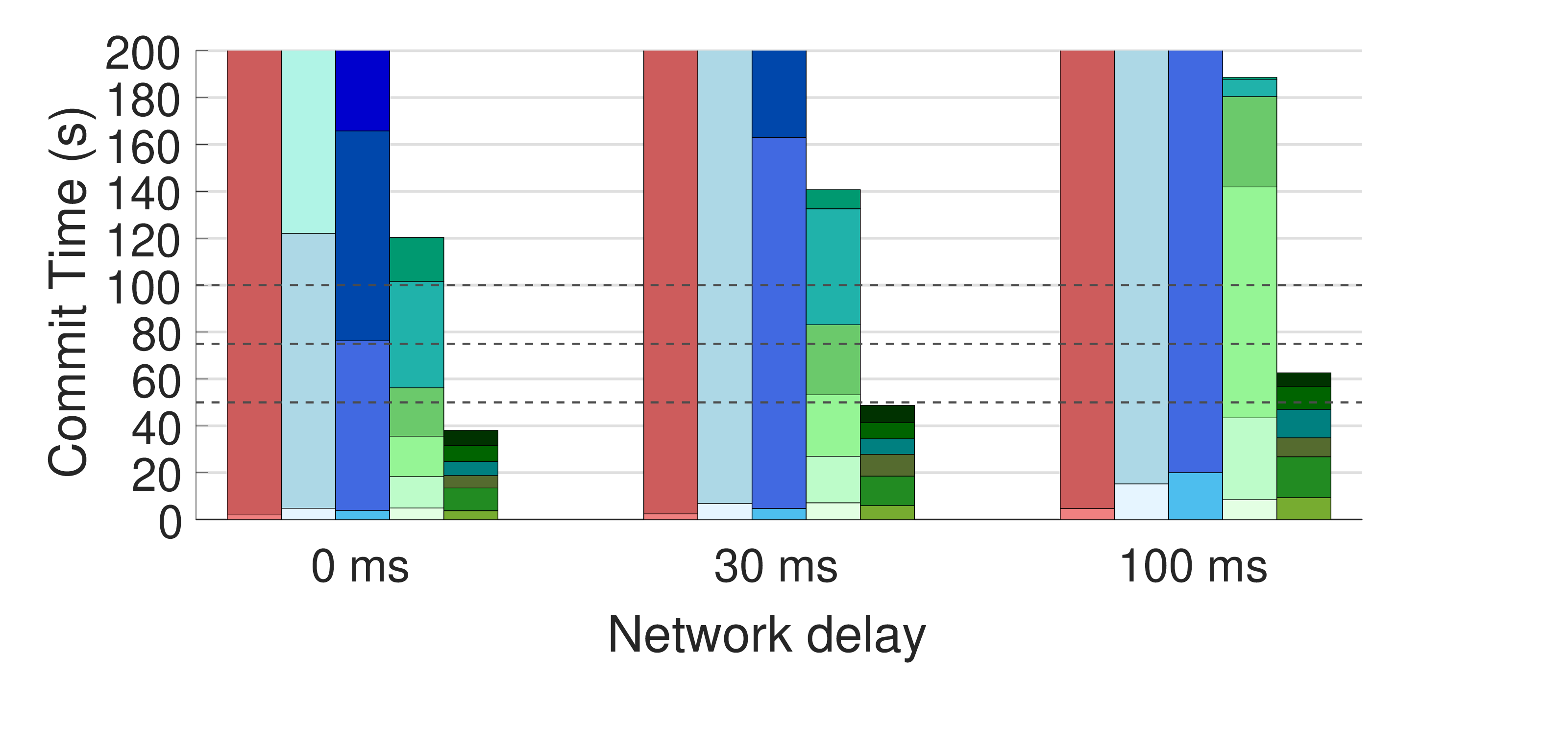}
\caption{Impact of network delay ($10$ servers, $10,000$ el/s).}
\label{subfig:impactofnetworkdelaysr10000committime}
\end{subfigure}
\hfill
\begin{subfigure}[b]{0.45\textwidth}
\centering
\includegraphics[trim={1cm 0cm 1cm 1cm},clip,width=\textwidth]{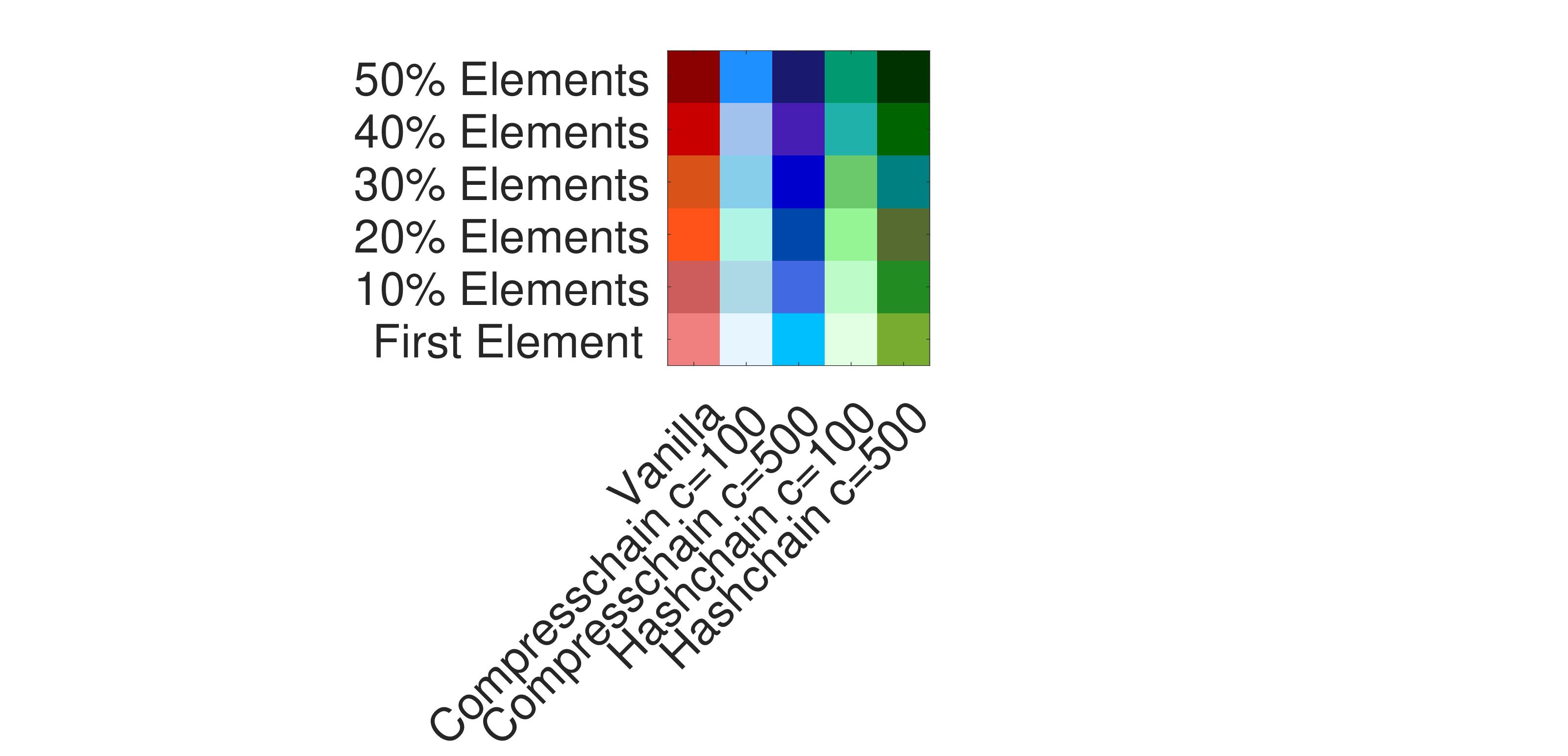}
\caption{Color codes}
\label{subfig:legendcommittimelarge}
\end{subfigure}
\caption{Commit times observed under different scenarios. The base scenario has $10$ servers, a sending rate of $10,000$ el/s, and no ($0$) network delay.}
\label{fig:impactofnofserverscommittime}
\end{figure*}

\section{Extension to a fully functional blockchain} \label{sec:extionsion_to_a_fully_functional_blockchain}
While the algorithms proposed are meant to implement a \setchain
object, in which elements inside an epoch have no order, in their
implementation, they impose an order inside each epoch (e.g., to hash
the epoch elements consistently).
%
%
This means that the proposed algorithms can be easily extended to
implement a blockchain, similarly as how Hyperledger Fabric or RedBelly
do:
%
(1) When adding elements to the \setchain and creating epochs, each transaction can be optimistically validated by itself (independently
  of all other transactions, that is, in parallel), ignoring its semantics.
%
(2) After each epoch is consolidated and its transactions ordered, the effect of its transactions can be computed (sequentially) in its
  actual final position. If a transaction is determined to be invalid it is marked as void.
%
Observe that extending \setchain to a blockchain can present a trade-off between decentralization and scalability.
As transaction execution must be done sequentially in an epoch, large epochs
may require large computational resources.
To ensure that less powerful servers can maintain synchronization with
the blockchain, it may be required
to limit epoch sizes, similar to Ethereum's block size
limitations~\cite{ethereumblocksize}.
\section{Limitations and Future Work} \label{sec:limitationandfuturework}
To the best of our knowledge, there is currently no blockchain benchmarking tool that allows for a direct comparison between~\setchain and state-of-the-art blockchain systems. We considered Diablo \cite{diablo}, a benchmarking suite focused on evaluating decentralized applications (dApps), but it does not suit our needs as~\setchain does not currently support dApps.
There are several potential directions in which future work could
evolve. One would be to run the experiments in a more distributed
environment to 
understand the scalability limits and ensure the system can handle
real-world demands. This would also include understanding the
tradeoffs between system parameters.
While using an underlying block-based ledger service, like Comet-BFT,
simplifies the algorithms, we have observed that it may also be a
bottleneck. We will explore if replacing this service with something
lighter, like a set consensus service \cite{redbelly} or Malachite \cite{malachite}, can increase performance and
scalability.
%
Another interesting work, as highlighted in Section \ref{subsubsec:pushing_hashchain_limits}, would be to implement a more efficient hash-reversal technique that could further improve Hashchain's performance.
Another promising direction would
be to design more efficient APIs for \setchain, and implement specific
DeFi applications that can use \setchain as the underlying decentralized infrastructure.
\end{document}